\documentclass[a4paper,UKenglish,11pt]{article}
\usepackage[bottom=2.5cm,top=2.5cm,left=2.5cm,right=2.5cm]{geometry}

\newcommand{\resultspath}{.}

\usepackage[utf8]{inputenc} 
\usepackage[T1]{fontenc}    
\usepackage{hyperref}       
\usepackage{url}            
\usepackage{booktabs}       
\usepackage{amsfonts}       
\usepackage{nicefrac}       

\usepackage{amsmath,amssymb,amsthm}

\usepackage{graphicx}
\usepackage{pgfplots}
\usepackage{filecontents}

\usepackage{tikz}
\usetikzlibrary{shapes}	

\RequirePackage{algorithm}
\RequirePackage{algorithmic}

\usepackage{thm-restate}

\newcommand{\REAL}{\ensuremath{\mathbb R}}
\DeclareMathOperator{\cost}{cost}
\DeclareMathOperator{\costf}{faircost}

\DeclareMathOperator{\costc}{colcost}

\newtheorem{theorem}{Theorem}
\newtheorem{definition}[theorem]{Definition}
\newtheorem{lemma}[theorem]{Lemma}
\newtheorem{question}[theorem]{Question}
\newtheorem{observation}[theorem]{Observation}
\newtheorem{proposition}[theorem]{Proposition}
\newtheorem{corollary}[theorem]{Corollary}

\begin{document}
\title{Fair Coresets and Streaming Algorithms for Fair k-means} 
\author{
  Melanie Schmidt\thanks{Department of Computer Science, Rheinische Friedrich-Wilhelms-Universität Bonn, 53111 Bonn, Germany, \texttt{melanieschmidt@uni-bonn.de}} \and Chris Schwiegelshohn\thanks{Department of Computer, Control and Management Engineering Antonio Ruberti,	Sapienza University of Rome, 00185 Roma RM, Italy,	\texttt{schwiegelshohn@diag.uniroma1.it}} \and Christian Sohler\thanks{Department of Computer Science,	TU Dortmund, 44227 Dortmund, Germany, \texttt{christian.sohler@tu-dortmund.de}}
}
\maketitle

\begin{abstract}
We study fair clustering problems as proposed by Chierichetti et al.~\cite{CKLV17}. Here, points have a sensitive attribute and all clusters in the solution are required to be balanced with respect to it (to counteract any form of data-inherent bias).
Previous algorithms for fair clustering do not scale well. 

We show how to model and compute so-called coresets for fair clustering problems, which can be used to significantly reduce the input data size.
We prove that the coresets are composable~\cite{IMMM14} and show how to compute them in a streaming setting.
Furthermore, we propose a variant of Lloyd's algorithm that computes fair clusterings and extend it to a fair $k$-means++ clustering algorithm. We implement these algorithms and provide empirical evidence that the combination of our approximation algorithms and the coreset construction yields a scalable algorithm for fair $k$-means clustering.
\end{abstract}

\thispagestyle{empty}
\newpage
\setcounter{page}{1}

\section{Introduction}

\begin{quote}
\emph{Our challenge is to support growth in the 
beneficial use of big data while ensuring that it 
does not create unintended discriminatory 
consequences.} \\
(Executive Office of the President, 2016 \cite{executive2016big})
\end{quote}

As the use of machine learning methods becomes more and more common in many areas of daily life ranging from automatic display of advertisements on webpages to mortgage approvals, we are faced with the question whether the decisions made by these automatic systems are \emph{fair}, i.e. free of biases by race, gender or other sensitive attributes. 
While at first glance it seems that replacing human decisions by algorithms will remove any kind of bias as algorithms will only decide based on the underlying data, the problem is that the training data may contain all sorts of biases. 
As a result, the outcome of an automated decision process may still contain these biases.

Recent findings in algorithmically generated results strengthen this concern. 
For example, it has been discovered that the COMPAS software that is used to predict the probability of recidivism is much more likely to assign an incorrect high risk score to a black defendant and low risk scores to a white defendant \cite{ALMK16}. 
This raises the general question how we can guarantee fairness in algorithms. 

This questions comes with several challenges. 
The first challenge is to formally define the concept of fairness. 
And indeed, it turns out that there are several ways to define \emph{fairness} which result in different optimal solutions \cite{CPFGH17}, and it has recently been shown that they cannot be achieved simultanuously unless the data has some very special (unlikely) structure \cite{KMR17}. 

In this paper we build upon the recent work by Chiericetti et al.~\cite{CKLV17} and consider fairness of clustering algorithms using the concept of \emph{disparate impact}, which is a notion of (un)fairness introduced to computer science by Feldman et. al.~\cite{FFMSV15}. 
Disparate impact essentially means that the result of a machine learning task does not correlate strongly with sensitive attributes such as gender, race sexual or religious orientation.
More formally and illustrated on the case of a single binary sensitive attribute $X$ and cluster variable $C$, a clustering algorithm does not show disparate impact if it satisfies the $p\%$ rule (a typical value for $p$ is $0.8$) stating that $\frac{\Pr\{C=i | X=0\}}{\Pr\{C=i | X=1\}} \le p.$
If we assume that both attribute values appear with the same frequency, then by Bayes Theorem 
the above translates to having at most $p\%$ points with a specific attribute value in each cluster. 

Chierichetti et. al. model fairness based on the disparate impact model in the following way.
They assume that every point has one of two colors (red or blue). 
If a set of points $C$ has $r_C$ red and $b_C$ blue points, then they define its \emph{balance} to be $\min(\frac{r_C}{b_C},\frac{b_C}{r_C})$. 
The overall balance of a clustering is then defined as the minimum balance of any cluster in it. 
Clusterings are then considered fair if their overall balance is close to $1/2$.


An algorithm ensuring fairness has to proceed with care; as mentioned before an algorithm that 
obliviously optimizes an objective function may retain biases inherent in the training set.
Chierichetti et al. avoid this by 
 identifying a set of fair micro-clusters via a suitably chosen perfect matching and running the subsequent optimization on the microclusters. 
This clever technique has the benefit of always computing a fair clustering, as the union of fair micro clusters is necessarily also fair.
However, the min-cost perfect matching is computationally expensive, and it needs random access to the data, which may be undesirable. 
This raises the following question:
\begin{question}
Is is possible to perform a fair data analysis efficiently, even when the size of the data set renders random-access unfeasible?
\end{question} 

\paragraph{Our contribution}
We address the issue of scaling algorithms by investigating \emph{coresets} for fair clustering problems, specifically for $k$-means. 
Given an input set $P$ in $d$ dimensional Euclidean space, the problem consists of finding a set of $k$ centers $c_1,\dots,c_k$ and a partition of $P$ into $k$ sets $C_1,\dots,C_k$ such that $\sum_{i=1}^k \sum_{p \in C_i} \|p-c_i\|_2^2$ is minimized.

Roughly speaking, a coreset is a summary of a point set which has the property that it approximates the cost function well for any possible candidate solution. The notion was proposed by Har-Peled and Mazumdar~\cite{HPM04} and has since recieved a wide range of attention, see~\cite{braverman2016new,FL11,FMS07,FGSSS13,FS05,LS10} for clustering. 
Coresets for geometric clustering are usually composable, meaning that if $S_1$ is a coreset for $P_1$ and $S_2$ is a coreset for $P_2$, then $S_1\cup S_2$ is a coreset for $P_1\cup P_2$~\cite{IMMM14}. Composability is arguably the main appeal of coresets; it enables an easy reduction from coreset constructions to streaming and distributed algorithms which scale to big data. 
Dealing with fair clustering, composability is not obvious. 
In this work, we initiate the study of fair coresets and their algorithmic implications:

$\bullet$ The standard coreset definition does not satisfy composability for fair clustering problems. We propose an alternative definition tailored to fair clustering problems and show that this new definition satisfies composability. The definition straightforwardly generalizes to having $\ell$ color classes and we show how a suitable coreset (of size $O(\ell k \log n \epsilon^{d-1})$ for constant $d$) may be computed efficiently.

$\bullet$ We provide different approximation algorithms for fair $k$-means clustering with two colors, including a variant of Lloyd's algorithm, two algorithms based on the approach of Chierichetti et al.~\cite{CKLV17}, and a $(1+\epsilon)$-approximation for cases with costant $k$. 

$\bullet$ We empirically evaluate the approximation algorithms and our coreset approach.  
In particular, we demonstrate empirically how coresets enable scalable fair clustering algorithms and also allow us to improve the solution quality by using better yet slower algorithms.

\paragraph{Additional related work}

The research on fairness in machine learning follows two main directions. One is to find proper definitions
of fairness. There are many different definitions available including statistical parity~\cite{TRT11},
disparate impact~\cite{FFMSV15}, disparate mistreatment~\cite{ZVG17} and many others, e.g.~\cite{BHJKR17,HPS16}. For an overview see the recent survey \cite{BHJKR17}.
Furthermore, the effects of different definitions of 
fairness and their relations have been studied in~\cite{C16,CPFGH17,KMR17}. A notion for individual fairness has been developed in~\cite{DHPRZ12}.
The other direction is the development of algorithms for fair machine learning tasks. Here the goal is to
develop new algorithms that solve learning tasks in such a way that the result satisfies a given fairness 
condition. Examples include~\cite{CKLV17,HPS16,ZVG17}. 
The closest result to our work is the above described paper by Chierichetti et. al. ~\cite{CKLV17}.

Polynomial-time approximation schemes for $k$-means were e.g. developed in~\cite{braverman2016new,FL11,KSS10}, assuming that $k$ is a constant. If $d$ is a constant,~\cite{CAKM16,FRS16} give a PTAS. If $k$ and $d$ are arbitrary, then the problem is APX-hard~\cite{AwasthiCKS15,LeeSW17}.
Lloyd's algorithm~\cite{L57} is an old but very popular local search algorithm for the $k$-means problem which can converge to arbitrarlity bad solutions. By using $k$-means++ seeding~\cite{AV07} as initialization, one can guarantee that the computed solution is a $O(\log k)$-approximation.

Chierichetti et al.~\cite{CKLV17} develop approximation algorithms for fair $k$-center and $k$-median with two colors. 
This approach was further improved by Backurs et al.~\cite{BIOSVW19}, who proposed an algorithm to speed up the fairlet computation.
R\"osner and Schmidt~\cite{RS18} extend their definition to multiple colors and develop an approximation algorithm for $k$-center. 
Bercea et al.~\cite{BGKKRSS18} develop an even more generalized notion and provide bicriteria approximations for fair variants of $k$-center, $k$-median and also $k$-means. For $k$-center, they provide a true $6$-approximation.
Very recently, Kleindessner et. al.~\cite{KAM19} proposed a linear-time $2$-approximation for fair $k$-center. This algorithm is not in the streaming setting, but still faster then previously existing approaches for fair clustering.

The fair $k$-means problem can also be viewed as a $k$-means clustering problem with size constraints. 
Ding and Xu \cite{DX15} showed how to compute an exponential sized list of candidate solutions for any of a large class of constrained clustering problems. Their result was improved by Bhattacharya et. al.~\cite{BJK18}.

In addition to the above cited coreset constructions, coresets for $k$-means in particular have also been studied empirically~\cite{AMRSLS12,AJM09,FGSSS13,FS08b,KMNPSW02}. 
Dimensionality reductions for $k$-means are for example proposed in~\cite{BZD10,BZMD15,CohenEMMP15,FeldmanSS13}. In particular\cite{CohenEMMP15,FeldmanSS13} show that any input to the $k$-means problem can be reduced to $\lceil k/\epsilon\rceil$ dimensions by using singular value decomposition (SVD) while distorting the cost function by no more than an $\epsilon$-factor. Furthermore, \cite{CohenEMMP15} also study random projection based dimensionality reductions. While SVD based reductions result in a smaller size, random projections are more versatile. We discuss the work of Cohen et. al. in more detail 

\subsection{Preliminaries}
We use $P\subseteq \REAL^d$ to denote a set of $n$ points in the $d$-dimensional vector space $\REAL^d$. 
The Euclidean distance between two points $p,q \in \REAL^d$ is denoted as $\|p-q\|$. 
The goal of clustering is to find a partition of an input point set $P$ into subsets of 'similar' points called clusters. 
In $k$-means clustering we formulate the problem as an optimization problem. 
The integer $k$ denotes the number of clusters. 
Each cluster has a center and the cost of a cluster is the sum of squared Euclidean distances to this center. 
Thus, the problem can be described as finding a set $C= \{c_1,\dots,c_k\}$ and corresponding clusters $C_1,\dots,C_k$ such that $\cost(P,C)=\sum_{i=1}^k \sum_{p\in C_i} \|p-c_i\|^2$ is minimized.

%
It is easy to see that in an optimal (non-fair) clustering each point $p$ is contained in the set $C_i$ such that $\|p-c_i\|^2$ is minimized.
The above definition can be easily extended to non-negatively and integer weighted point sets by treating the weight as a multiplicity of a point. We denote the $k$-means cost of a set $S$ weighted with $w$  and center set $C$ as $\cost_w(S,C)$. Finally, we recall that the best center for a cluster $C_i$ is its centroid $\mu(C_i):=\frac{1}{|C_i|}\sum_{p\in C_i}p$.
\begin{proposition}
\label{prop:zauberformel}
Given a point set $P\subset \mathbb{R}^d$ and a point $c\in\mathbb{R}^d$, the $1$-means cost of clustering $P$ with $c$ can be decomposed into
$\sum_{p\in P} \|p-c\|^2 = \sum_{p\in P} \|p-\mu(P)\|^2 + |P|\cdot \|\mu(P)-c\|^2$.
\end{proposition}
Next, we give the coreset definition for $k$-means as introduced by Har-Peled and Mazumdar.
\begin{definition}[Coreset~\cite{HPM04}]\label{def:coreset}
A set $S \subseteq \REAL^d$ together with non-negative weights $w:S \to \mathbb{N}$ is a $(k,\epsilon)$-coreset for a point set $P\subseteq \REAL^d$ with respect to the $k$-means clustering problem, if for every set $C
\subseteq \REAL^d$ of $k$ centers we have
$\cost_w(S,C) \in  (1\pm\epsilon) \cdot \cost(P,C)$.
\end{definition}

\paragraph{Fair clustering} We extend the definition of fairness from~\cite{CKLV17} to sensitive attributes with multiple possible values. As in~\cite{CKLV17}, we model the sensitive attribute by a color. 
Notice that we can model multiple sensitive attributes by assigning a different color to any combination of possible values of the sensitive attributes.
We further assume that the sensitive attributes are not among the point coordinates. Thus, our input set is a set $P \subseteq \REAL^d$ together with a coloring $c: P \rightarrow \{1,\dots,\ell\}$. 

We define $\xi(j)=|\{p\in P : c(p) = j\}| / |P|$ as the fraction that color $j$ has in the input point set. Then we call a clustering $C_1,\dots,C_k$ $(\alpha,\beta)$-fair, $0 < \alpha \le 1 \le \beta$, if for every cluster $C_i$ and every color class $j \in \{1,\dots,\ell\}$ we have  
\[
\alpha \cdot \xi(j) \le \frac{|\{p\in C_i : c(p) = j\}|}{|\{p\in C_i\}|} \le \beta \cdot  \xi(j) .
\]
For any set $C=\{c_1,\dots,c_k\}$ of $k$ centers we define 
$
\costf(P,C) 
$
to be the minimum of $\sum_{i=1}^k \sum_{p\in C_i} \|p-c_i\|^2$ where the minimum is taken over all $(\alpha,\beta)$-fair clusterings of $P$ into $C_1,\dots,C_k$.
The optimal $(\alpha,\beta)$-fair clustering $C'$ is the one with minimal $\costf(P,C')$.
Alternatively to $\xi(j)$, we could demand that the fraction of all colors is (roughly) $1/\ell$. However, notice that the best achievable fraction is $\xi(j)$. 
Thus $(\alpha,\beta)$-fairness is a strictly more general condition.
It is also arguably more meaningful if the data set itself is heavily imbalanced. Consider an instance where the blue points outnumber the red points by a factor of $100$. Then the disparity of impact is at least $0.01$. A $(1,1)$-fair clustering then is a clustering where all clusters achieve the best-possible ratio $0.01$. 

\section{Fair coresets and how to get them}\label{sec:faircoresets}
First, notice that the \emph{definition} of coresets as given in Definition~\ref{def:coreset} does not translate well to the realm of fair clustering. Assume we replace $\cost$ by $\costf$ in Definition~\ref{def:coreset}. Now consider Figure~\ref{fig:nocomposability}.
\begin{figure}
\centering
\scalebox{0.7}{
\begin{tikzpicture}[xscale=1.1,yscale=0.9,
bp/.style={draw=blue,fill=blue,minimum width=0.15cm,minimum height=0.15cm,inner sep=0cm},
rp/.style={draw=red,fill=red,diamond,minimum width=0.18cm,,minimum height=0.18cm,inner sep=0cm},
p/.style={draw=black,fill=black,circle,minimum width=0.1cm,minimum height=0.1cm,inner sep=0cm},
c/.style={draw=black,fill=white,circle,minimum width=0.1cm,minimum height=0.1cm,inner sep=0cm}]
\useasboundingbox (-0.7,-4) rectangle (17,1.6);
\begin{scope}
\node [rp] at (0,0) {};
\node [rp] at (0,0.25) {};
\node [bp] at (3.6,0) {};
\node [bp] at (3.6,0.25) {};
\node [rp] at (0,1) {};
\node [rp] at (0,1.25) {};
\node [bp] at (3.6,1) {};
\node [bp] at (3.6,1.25) {};
\draw [<->] (0.2,0.25) to node [fill=white] {$\Delta$} (3.4,0.25);
\draw [<->] (-0.2,-0.05) to (-0.2,0.3);
\node at (-0.5,0.12) {$\varepsilon$}; 
\draw [<->] (3.8,-0.05) to (3.8,1.3);
\node at (4.2,0.6) {$\varepsilon \Delta$}; 
\node at (1.8,-0.8) {point set $P_1$};
\end{scope}
\begin{scope}[yshift=-3cm]
\node (y1) [bp] at (0.2,0) {};
\node (y2) [bp] at (0.2,0.25) {};
\node (y3) [rp] at (3.4,0) {};
\node (y4) [rp] at (3.4,0.25) {};
\node (y5) [bp] at (0.2,1) {};
\node (y6) [bp] at (0.2,1.25) {};
\node (y7) [rp] at (3.4,1) {};
\node (y8) [rp] at (3.4,1.25) {};
\node at (1.8,-0.8) {point set $P_2$};
\end{scope}
\begin{scope}[xshift=6.5cm]
\node [rp,label=above:{$4$}] at (0,0.625) {};
\node [bp,label=above:{$4$}] at (3.6,0.625) {};
\node at (1.8,-0.8) {\lq coreset\rq\ $S_1$};
\end{scope}
\begin{scope}[xshift=6.5cm,yshift=-3cm]
\node [bp,label=above:{$4$}] at (0.2,0.625) {};
\node [rp,label=above:{$4$}] at (3.4,0.625) {};
\node at (1.8,-0.8) {\lq coreset\rq\ $S_2$};
\end{scope}
\begin{scope}[xshift=12.5cm]
\node (x1) [rp] at (0,0) {};
\node (x2) [rp] at (0,0.25) {};
\node (x3) [bp] at (3.6,0) {};
\node (x4) [bp] at (3.6,0.25) {};
\node (x5) [rp] at (0,1) {};
\node (x6) [rp] at (0,1.25) {};
\node (x7) [bp] at (3.6,1) {};
\node (x8) [bp] at (3.6,1.25) {};
\node (y1) [bp] at (0.2,0) {};
\node (y2) [bp] at (0.2,0.25) {};
\node (y3) [rp] at (3.4,0) {};
\node (y4) [rp] at (3.4,0.25) {};
\node (y5) [bp] at (0.2,1) {};
\node (y6) [bp] at (0.2,1.25) {};
\node (y7) [rp] at (3.4,1) {};
\node (y8) [rp] at (3.4,1.25) {};
\node (c1) [c] at (0.1,0.125) {};
\node (c2) [c] at (3.5,0.125) {};
\node (c3) [c] at (0.1,1.125) {};
\node (c4) [c] at (3.5,1.125) {};
\node at (1.8,-0.8) {$\costf(P_1\cup P_2)\in \mathcal{O}(\epsilon)$};
\draw (x1) to (c1);
\draw (x2) to (c1);
\draw (x3) to (c2);
\draw (x4) to (c2);
\draw (x5) to (c3);
\draw (x6) to (c3);
\draw (x7) to (c4);
\draw (x8) to (c4);
\draw (y1) to (c1);
\draw (y2) to (c1);
\draw (y3) to (c2);
\draw (y4) to (c2);
\draw (y5) to (c3);
\draw (y6) to (c3);
\draw (y7) to (c4);
\draw (y8) to (c4);
\end{scope}
\begin{scope}[xshift=12.5cm,yshift=-3cm]
\node (x1) [rp,label={[label distance=0.15cm]right:{$4$}}] at (0,0.625) {};
\node (x2) [bp,label={[label distance=0.15cm]left:{$4$}}] at (3.6,0.625) {};
\node (x3) [bp,label={[label distance=0.15cm]left:{$4$}}] at (0.2,0.625) {};
\node (x4) [rp,label={[label distance=0.15cm]right:{$4$}}] at (3.4,0.625) {};
\node at (1.8,-0.8) {$\costf(S_1\cup S_2)\in \Omega(\epsilon\Delta)$};
\node (c1) [c] at (0.15,0.125) {};
\node (c2) [c] at (3.45,0.125) {};
\node (c3) [c] at (0.15,1.125) {};
\node (c4) [c] at (3.45,1.125) {};
\draw (x1) to (c1);
\draw (x2) to (c2);
\draw (x3) to (c1);
\draw (x4) to (c2);
\end{scope}
\end{tikzpicture}
}
\caption{A simple example of non-composable coresets.\label{fig:nocomposability}}
\end{figure}
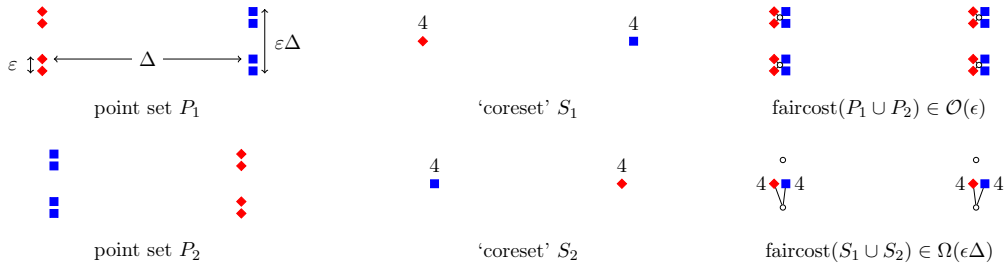

We see two point sets $P_1$ and $P_2$ with eight points each, which both have an optimum cost of $\Omega(\Delta)$. Replacing the four left and the four right points by one point induces an error of $\mathcal{O}(\epsilon \Delta)$, which is an $\mathcal{O}(\epsilon)$-fraction of the total cost.
Thus, the depicted sets $S_1$ and $S_2$ are coresets.
However, when we combine $P_1$ and $P_2$, then the optimum changes. The cost decreases dramatically to $\mathcal{O}(\epsilon)$. For the new optimal solution, $S_1\cup S_2$ still costs $\Omega(\epsilon \Delta)$, and the inequality in Definition~\ref{def:coreset} is no longer satisfied. 

We thus have to do a detour: We define a stronger, more complicated notion of coresets which regains the property of being composable. Then, we show that a special type of coreset constructions for $k$-means can be used to compute coresets that satisfy this stronger notion. It is an interesting open question to analyze whether it is possible to design sampling based coreset constructions that satisfy our notion of coresets for fair clustering.

For our detour, we need the following generalization of the standard $k$-means cost. 
A \emph{coloring constraint} for a set of $k$ cluster centers $C=\{c_1,\dots,c_k\}$ and a set of $\ell$ colors $\{1,\dots,\ell\}$ is a $k\times \ell$ matrix $K$. Given a point set $P$ with a coloring $c:P \rightarrow \{1,\dots,\ell\}$ we say that a partition of $P$ into sets $C_1,\dots, C_k$ satisfies $K$ if $|\{p\in C_i : c(p) =j\}| = K_{ij}$. The cost of the corresponding clustering
is $\displaystyle\sum_{i=1}^k \sum_{p\in C_i} \|p-c_i\|^2$ as before. 
Now we define the \emph{color-$k$-means cost} $\costc(P,K,C)$ to be the minimal cost of any clustering satisfying $K$. 
If no clustering satisfies $K$, $\costc(P,K,C):=\infty$.

Notice that we can prevent the bad example in Figure~\ref{fig:nocomposability} by using the color-$k$-means cost: If $\costc(P,K,C)$ is approximately preserved for the color constraints modeling that each cluster is either completely blue or completely red, then $S_1$ and $S_2$ are forbidden as a coresets. 

\begin{definition}
\label{def:faircoreset}
Let $P$ be a point set with coloring $c:P \rightarrow \{1,\dots,\ell\}$.
A non-negatively integer weighted set $S \subseteq \REAL^d$ with a coloring $c':S \rightarrow \{1,\dots,\ell\}$ is a $(k,\epsilon)$-coreset for $P$ for the $(\alpha,\beta)$-fair $k$-means 
clustering problem, if for every set $C \subseteq \REAL^d$ of $k$ centers and every coloring constraint $K$ we have
\[
\costc_w(S,K,C) \in (1\pm\epsilon) \cdot \costc(P,K,C),
\]
where in the computation of $\costc(S,K,C)$ we treat a point with weight $w$ as $w$ unweighted points and therefore
a point can be partially assigned to more than one cluster.
\end{definition}

Definition~\ref{def:faircoreset} demands that the cost is approximated for \emph{any} possible color constraint. This implies that it is approximated for those constraints we are interested in. Indeed, the fairness constraint can be modeled as a collection of color constraints. As an example for this, assume we have two colors and $k$ is also two; furthermore, assume that the input is perfectly balanced, i.e., the number of points of both colors is $n/2$, and that we want this to be true for both clusters as well. Say we have a center set $C=\{c_1,c_2\}$ and define $K^i$ by $K_{11}^i=i,K_{12}^i=i,K_{21}^i = \frac{n}{2}-i, K_{22}^i = \frac{n}{2}-i$, i.e., $K^i$ assigns $i$ points of each color to $c_1$ and the rest to $c_2$. The feasible assignments for the fairness constraint are exactly those assignments that are legal for exactly one of the color constraints $K^i$, $i\in \{0,\ldots,\frac{n}{2}\}$. So since a coreset according to Definition~\ref{def:faircoreset} approximates $\costc(P,C,K^i)$ for all $i$, it in particular approximately preserves the cost of any fair clustering.
This also works in the general case: We can model the $(\alpha,\beta)$-fair constraint as a collection of color constraints (and basically any other fairness notion based on the fraction of the colors in the clusters as well). 

\begin{restatable}{proposition}{colcostmodelsthings}\label{lem:colcostmodels}
Given a center set $C$, $|C|=k$, the assignment restriction to be 
$(\alpha,\beta)$-fair can be modeled as a collection of coloring constraints.
\end{restatable}
\begin{proof}
Recall $\xi(j)=\frac{|\{p\in P~:~c(p)=j\}|}{|P|}$.
Let $C=\{C_1,\ldots ,C_k\}$ be a clustering and let $K$ be the coloring constraint matrix induced by $C$. We observe that the $i$th row sums up to $|C_i|$ and the $j$th column sums up to $|\{p\in P~:~c(p)=j\}|$.
Then $C$ is $(\alpha,\beta)$-fair if and only if $\alpha\cdot \xi(j)\leq \frac{|\{p\in C_i~:~c(p)=j\}|}{|C_i|}=\frac{K_{i,j}}{\sum_{h=1}^k K_{i,h}} \leq \beta\cdot \xi(j)$ for all $i\in \{1,\ldots,k\}$ and $j\in\{1,\ldots,\ell\}$.
\end{proof}

The main advantage of Definition~\ref{def:faircoreset} is that it satisfies composability. 
The main idea is that for any coloring constraint $K$, any clustering satisfying $K$ induces specific color constraints $K_1$ and $K_2$ for $P_1$ and $P_2$; and for these, the coresets $S_1$ and $S_2$ also have to satisfy the coreset property. We can thus proof the coreset property for $S$ and $K$ by composing the guarantees for $S_1$ and $S_2$ on $K_1$ and $K_2$.

\begin{restatable}[Composability]{lemma}{coresetcomposable}
Let $P_1, P_2 \subset \mathbb{R}^d$ be point sets.
Let $S_1$, $w_1$, $c_1$ be a $(k,\epsilon)$-coreset for $P_1$ and let $S_2$, $w_2$, $c_2$ be a $(k,\epsilon)$-coreset for $P_2$ (both satisfying Definition~\ref{def:faircoreset}). Let $S=S_1\cup S_2$ and concatenate $w_1, w_2$ and $c_1, c_2$ accordingly to obtain weights $w$ and colors $c$ for $S$. 
Then $S$, $w$, $c$ is a coreset for $P=P_1 \cup P_2$ satisfying Definition~\ref{def:faircoreset}.
\end{restatable}
\begin{proof}
Let $C = \{c_1,\ldots,c_k\} \subset \mathbb{R}^d$ be an arbitrary set of centers, and let $K\in\mathbb{N}^{k \times \ell}$ be a an arbitrary coloring constraint for $C$. We want to show that
\[
\costc_w(S,K,C) \in (1\pm \epsilon)\costc(P,K,C).
\]
Let $\gamma : P \to C$ be an assignment that minimizes the assignment cost among all assignments that satisfy $K$, implying that $\costc(P,K,C) = \sum_{p \in P} ||x-\gamma(x)||^2$. Since $\gamma$ satisfies $K$, the number of points of color $j$ assigned to each center $c_i \in C$ is exactly $K_{ij}$. We split $K$ into two matrices $K_1$ and $K_2$ with $K=K_1+K_2$ by counting the number of points of each color at each center which belong to $P_1$ and $P_2$, respectively. In the same fashion, we define two mappings $\gamma_1 : P_1 \to C$ and $\gamma_2: P_2 \to C$ with $\gamma_1(p)=\gamma(p)$ for all $p \in P_1$ and $\gamma_2(p)=\gamma(p)$ for all $p \in P_2$. 

Now we argue that $\costc(P,C,K)=\costc(P_1,C,K_1) + \costc(P_2,C,K_2)$. Firstly, we observe that $\costc(P,C,K)\le\costc(P_1,C,K_1) + \costc(P_2,C,K_2)$ since $\gamma_1$ and $\gamma_2$ are legal assignments for the color constraint $K_1$ and $K_2$, respectively, and they induce exactly the same point-wise cost as $\gamma$. Secondly, we argue that there cannot be cheaper assignments for $K_1$ and $K_2$. Assume there where an assignment $\gamma_1'$ with $\sum_{p \in P_1} ||x-\gamma_1'(x)||^2 < \costc(P_1,C,K_1)$. Then we could immediately adjust $\gamma$ to be identical to $\gamma_1'$ on the points in $P_1$ instead of $\gamma_1$, and this would reduce the cost; a contradiction to the optimality of $\gamma$. The same argument holds for $\gamma_2$. Thus, $\costc(P,C,K)=\costc(P_1,C,K_1) + \costc(P_2,C,K_2)$ is indeed true.

Now since $S_1$, $w_1$, $c_1$ is a coreset for $P_1$ and $S_2$, $w_2$, $c_2$ is a coreset for $P_2$, they have to approximate $\costc(P_1,C,K_1)$ and $\costc(P_2,C,K_2)$ well. We get from this that
\begin{align*}
 & \costc_w(S_1,C,K_1) + \costc_w(S_2,C,K_2) \\
&\in (1\pm \epsilon)\cdot \costc(P_1,C,K_1) + (1\pm \epsilon)\cdot \costc(P_2,C,K_2) \\
&\in (1\pm \epsilon) \cdot \costc(P,C,K).
\end{align*}
Observe that $\costc_w(S,C,K) \le \costc_w(S_1,C,K_1) + \costc_w(S_2,C,K_2)$ since we can concatenate the optimal assignments for $S_1$ and $S_2$ to get an assignment for $S$. Thus, $\costc_w(S,C,K) \le (1+\epsilon)\cdot\costc(P,C,K)$.
It remains to show that
$\costc_w(S,C,K) \ge (1-\epsilon)\cdot\costc(P,C,K)$. 

Let $\gamma':S\to C$ be an assignment that satisfies $K$ and has cost $\costc_w(S,C,K)$ (for simplicity, we treat $S$ as if it were expanded by adding multiple copies of each weighted point; recall that we allow weights to be split up for $S$). Let $\gamma_1':P_1 \to C$ and $\gamma_2':P_2\to C$ be the result of translating $\gamma'$ to $P_1$ and $P_2$, and split $K$ into $K_1'$ and $K_2'$ according to $\gamma'$ as we did above. Then 
$\costc_w(S,C,K)=\costc_w(S_1,C,K_1')+\costc_w(S_2,C,K_2')$ by the same argumentation as above. Furthermore, 
\begin{align*}
\costc_w(S,C,K) =& \costc_w(S_1,C,K_1')+\costc_w(S_2,C,K_2') \\
\ge& (1- \epsilon) \costc_w(P_1,C,K_1')+(1- \epsilon)\costc_w(P_2,C,K_2')\\
\ge& (1- \epsilon) \costc(P,C,K).
\end{align*}
where the first inequality holds by the coreset property and the second is true since we can also use $\gamma'$ to cluster $P$, implying that $ \costc_w(P,C,K) \le  \costc_w(P_1,C,K_1')+\costc_w(P_2,C,K_2')$. That completes the proof.
\end{proof}

We have thus achieved our goal of finding a suitable \emph{definition} of coresets for fair clustering. Now the question is whether we can actually compute sets which satisfy the rather strong 
Definition~\ref{def:faircoreset}. Luckily, we can show that a special class of coreset constructions for $k$-means can be adjusted to work for our purpose. A \emph{coreset construction for $k$-means} is an algorithm that takes a point set $P$ as input and computes a summary $S$ with integer weights that satisfies Definition~\ref{def:coreset}. 

We say that a coreset construction is \emph{movement-based} if
\begin{itemize}
\item all weights $w(p), p \in S$ are integers
\item there exists a mapping $\pi : P \rightarrow S$ with $\sigma^{-1}(p) = w(p)$ for all $p \in S$ which satisfies that 
$\sum\limits_{x \in P} || x - \pi(x) ||^2 \leq \dfrac{\varepsilon^2}{16} \cdot \text{OPT}_k$, where $\text{OPT}_k=\min\limits_{C\subset \mathbb{R}^d,|C|=k} \cost(P,C)$. 
\end{itemize} 
Movement-based coreset constructions compute a coreset by \lq moving\rq\ points to common places at little cost, and then replacing heaps of points by weighted points. Examples for movement-based coreset constructions are~\cite{FGSSS13,FS05,HPM04}.
Now the crucial observation is that we can turn any movement-based coreset construction for $k$-means into an algorithm that computes coresets for fair $k$-means satisfying Definition~\ref{def:faircoreset}. 
The main idea is to run $ALG$ to move all points in $P$ to common locations, and then to replace all points \emph{of the same color} at the same location by one coreset point. This may result in up to $\ell$ points for every location, i.e., the final coreset result may be larger than its colorless counterpart by a factor of at most $\ell$. The rest of the proof then shows that Definition~\ref{def:faircoreset} is indeed true, following the lines of movement-based coreset construction proofs. 

\begin{restatable}{theorem}{thmmovbased}\label{thm:movbased}
Let $ALG$ be a movement-based coreset construction for the $k$-means problem.
Assume that given the output $P \in \mathcal{R}^d$, $k$ and $\epsilon$, 
the size of the coreset that $ALG$ computes is bounded by $f(|P|,d,k,\epsilon)$. 
Then we can construct an algorithm $ALG'$ which constructs a set $S'$ that satisfies Definition~\ref{def:faircoreset}. The size of this set is bounded by $\ell \cdot  f(|P|,d,k,\epsilon)$, where $\ell$ is the number of colors. 
\end{restatable}
\begin{proof}
For any $P$, $ALG$ gives us a set $S$ and a non-negative weight function $w$ such that Definition~\ref{def:coreset} is true, i.e., 
\begin{equation}
\cost_w(S,C) \in (1\pm\epsilon)\cost(P,C)\label{eq:coreset:inproof:movement}
\end{equation}
holds for all set of centers $C$ with $|C|=k$. Since $ALG$ is movement-based, the weights are integer; and there exists a mapping $\pi : P \to S$, such that at most $w(p)$ points from $P$ are mapped to any point $p \in S$, and such that 
\begin{equation}
\sum\limits_{x \in P} || x - \pi(x) ||^2 \leq \dfrac{\varepsilon^2}{16} \cdot \text{OPT}_k
\label{eq:strongercoresetp:inproof:movement}
\end{equation}
is true. Statement~\eqref{eq:strongercoresetp:inproof:movement} is stronger than~\eqref{eq:coreset:inproof:movement}, and we will only need~\eqref{eq:strongercoresetp:inproof:movement} for our proof. We will, however, need the mapping $\pi$ to construct $ALG'$.
Usually, the mapping will be at least implicitly computed by $ALG$. If not or if outputting this information from $ALG$ is cumbersome, we do the following. We assign every point in $P$ to its closest point in $S$. The resulting mapping has to satisfy~\ref{eq:strongercoresetp:inproof:movement}, since the distance of any point to its closest point in $S$ can only be smaller than in any given assignment. We may now assign more than $w(p)$ points to $S$. We resolve this by simply changing the weights of the points in $S$ to match our mapping. Since we now have $S$, $w$ and $\pi$ satisfying~\eqref{eq:strongercoresetp:inproof:movement}, we can proceed as if $ALG$ had given a mapping to us.

Now we do what movement-based coreset constructions do internally as well: We consolidate all points that share the same location. However, since they may not all be of the same color, we possibly put multiple (at most $\ell$) copies of any point in $S$ into our coreset $S'$. More precisely, for every $p \in S$, we count the number $n_{p,i}$ of points of color $i$. If $n_{p,i}$ is at least one, then we put $p$ into $S'$ with color $i$ and weight $n_{p,i}$. The size of $S'$ is thus at most $\ell \cdot f(|P|,d,k,\epsilon)$.

The proof that $S'$ satisfies~\ref{def:faircoreset} is now close to the proof that movement-based coreset constructions work. To execute it, we imagine $S'$ in its expanded form (where every point $p$ is replaced by $n_{p,i}$ points of color $i$. We call this expanded version $P'$. Notice that $\cost(P',C)=\cost_w(S',C)$ for all $C \subset \mathbb{R}^d$. We only need $P'$ for the analysis. Notice that $\pi$ can now be interpreted as a bijective mapping between $P$ and $P'$ and this is how we will use it.

Let $C$ be an arbitrary center set with $|C|=k$ and let $K$ be an arbitrary coloring constraint.
We want to show that $\costc(P',K,C) \in (1\pm\epsilon) \cdot \costc(P,K,C)$. 
If no assignment satisfies $K$, then $\costc(P,K,C)$ is infinity, and there is nothing to show.
Otherwise, fix an arbitrary optimal assignment $\gamma:P\to C$ of the points in $P$ to $C$ among all assignments that satisfy $K$. Notice that $\gamma$ and $\pi$ are different assignments with different purposes; $\gamma$ assigning a point in $P$ to its center, and $\pi$ assigning each point in $P$ to its moved version in $P'$.

We let $v_c(x) := ||x-\gamma(x)||$ be the distance between $x\in P$ and the center its assigned to. Let $v_c$ be the $|P|$-dimensional vector consisting of all $v_c(x)$ (in arbitrary order). Then we have \[\costc(P,C,K) = \sum_{x \in P} ||x-\gamma(x)||^2 = \sum_{x \in P} v_c(x)^2 = ||v_c||^2.\]

Furthermore, we set $v_p(x) = ||\pi(x) - x||$ and let $v_p$ be the $|P|$-dimensional vector of all $v_p(x)$ (ordered in the same way as $v_c$). We have $\sum_{x \in P} ||\pi(x)-x||^2 \le \dfrac{\varepsilon^2}{16} \cdot OPT_k$ by our preconditions.

Now we want to find an upper bound on $\costc(P',C,K)$. Since we only need an upper bound, we can use $\gamma$ for assigning the points in $P'$ to $C$. We already know that $\gamma$ satisfies $K$ for the points in $P$; and the points in $P'$ are only moved versions of the points in $P$. 
We use this and then apply the triangle inequality:
\begin{align*}
\costc(P',C,K) \le& \sum_{y \in P'} ||y-\gamma(\pi^{-1}(y))||^2= \sum_{x \in P} ||\gamma(x)-\pi(x)||^2\\
\le& \sum_{x \in P} (||\gamma(x)-x||+||x-\pi(x)||)^2
= \sum_{x \in P} (v_c(x)+v_p(x))^2
= ||v_c+c_p||^2.
\end{align*}
Now we can apply the triangle inequality to the vector $v_c+v_p$ to get  $||v_c+v_p|| \le ||v_c||+||v_p|| \le \sqrt{\costc(P,C,K)} + \sqrt{\frac{\varepsilon^2}{16} }\cdot OPT_k$.
So we know that
\begin{align*}
\costc(P',C,K)\le||v_c+v_p||^2 & \le \costc(P,C,K) + \frac{\varepsilon^2}{16} \cdot OPT_k 
\\& \quad + 2 \sqrt{\costc(P,C,K)}\cdot \sqrt{\frac{\varepsilon^2}{16}\cdot OPT_k}\\
& \le \costc(P,C,K) + \frac{\varepsilon^2}{16} \cdot OPT_k \\
& \quad + \frac{\varepsilon}{2} \cdot \costc(P,C,K)\\
& < (1+\epsilon) \cdot \costc(P,C,K).
\end{align*}

To obtain that also $\costc(P,C,K) \le (1+\epsilon)\cdot \costc(P',C,K)$, we observe that the above argumentation is symmetric in $P$ and $P'$. No argument used that $P$ is the original point set and $P'$ is the moved version. So we exchange the roles of $P$ and $P'$ to complete the proof.
\end{proof}

We can now apply Theorem~\ref{thm:movbased}. Movement-based constructions include the original paper due to Har-Peled and Mazumdar~\cite{HPM04} as well as the practically more efficient algorithm BICO~\cite{FGSSS13}. For more information on the idea of movement-based coreset constructions, see Section 3.1 in the survey~\cite{MS18}. For BICO in particular, Lemma 5.4.3 in~\cite{S14} gives a proof that the construction is movement-based. Using Theorem~\ref{thm:movbased} and Corollary~1 from~\cite{FGSSS13}, we then obtain the following.

\begin{corollary}
\label{cor:fairbico}
There is an algorithm in the insertion-only streaming model which computes a $(k,\epsilon)$-coreset for the fair $k$-means problem according to Definition~\ref{def:faircoreset}. The size of the coreset and the storage requirement of the algorithm is $m\in O(\ell \cdot k \cdot \log n \cdot \epsilon^{-d+2})$, where $\ell$ is the number of colors in the input, and where $d$ is assumed to be a constant.

The running time of the algorithm is $O(N(m)(n+\log(n\Delta)m))$, where $\Delta$ is the spread of the input points and $N(m)$ is the time to compute an (approximate) nearest neighbor. 
\end{corollary}

\section{Approximation algorithms for fair \texorpdfstring{$k$}{k}-means}

We give a full overview of algorithms for fair $k$-means clustering in Section~\ref{sec:appendix:algorithms}, but give an overiew here. 
Notice that while the previous section was for multiple colors, we now go back to the case with only two colors, assuming that exactly half of the input points are colored with each color, and demanding that this is true for all clusters in the clustering as well. We call this special case \emph{exactly balanced}.
We do this because for multiple colors, no true approximation algorithms are known, and there is indication that this problem might be very difficult (it is related to solving capacitated $k$-median/$k$-means, a notoriously difficult question). Notice that the coreset approach works for arbitrary $(\alpha,\beta)$-fair $k$-means.

For two colors, Chierichetti et. al. \cite{CKLV17} outline how to transfer approximation algorithms for clustering to the setting of fair clustering, but derive the algorithms only for $k$-center and $k$-median. The idea is to first compute a coarse clustering where the microclusters are called \emph{fairlets}, and then to cluster representatives of the fairlets to obtain the final clustering. The following algorithm extends their ideas to compute fairlets for $k$-means. 

\begin{algorithm}
\caption{Fairlet computation}
\label{alg:fairlets:mainpart}
\let\oldnl\nl
\newcommand{\nonl}{\renewcommand{\nl}{\let\nl\oldnl}}
1: Let $B$ be the blue points and $R$ be the red points\\
2: For any $b \in B$, $r \in R$, set $c(r,b) = ||r-b||^2/2$\\
3: Consider the complete bipartite graph $G$ on $B$ and $R$\\
4: Compute a min cost perfect matching $M$ on $G$\\
5: For each edge $(r,b)\in M$, add $\mu(\{r,b\})$ to $F$\\
6: Output $F$
\end{algorithm}

The idea of the algorithm is the following. In any optimal solution, the points can be paired into tuples of a blue and a red point which belong to the same optimal cluster. Clustering the $n/2 \ge k$ tuples with $n/2$ centers cannot be more expensive than the cost of the actual optimal $k$-means solution. Thus, we would ideally like to know the tuples and partition them into clusters. Since we cannot know the tuples, we instead compute a minimum cost perfect matching between the red and blue points, where the weight of an edge is the $1$-means cost of clustering the two points with an optimal center (this is always half their squared distance). The matching gives us tuples; these tuples are the fairlets. The centroid of each fairlet now serves as its representative.
The following theorem shows that clustering the representatives yields a good solution for the original problem. 

\begin{restatable}{theorem}{thmfairletskmeans}\label{thm:kmeansfairlets}
There is an algorithm that achieves the following. For any $P \subset \mathbb{R}^d$ which contains $|P|/2$ blue and $|P|/2$ red points, it computes a set of representatives $F \subset P$ of size $|P|/2$, such that an $\alpha$-approximate solution for the colorless $k$-means problem on $F$ yields a $(5.5\alpha+1)$-approximation for the fair $k$-means problem on $P$. 
\end{restatable}

Similarly to the fairlet computation, the problem of finding an \emph{fair assignment}, i.e., an cost-wise optimal assignment of points to \emph{given} centers which is fair, can be modeled as a matching problem. This algorithm, as well as algorithms for fairlet computation and fair assignment for \emph{weighted} points are specified in Section~\ref{sec:appendix:algorithms} (Algorithms~\ref{alg:fairassignment}, \ref{alg:fairlets:weighted}, \ref{alg:fairassignment:weighted}).

The fairlet computation and fair assignment give rise to the following algorithms, where we use $k$-means++ as the approximation algorithm for the unconstrained $k$-means problem. It consists of a clever sampling step called $D^2$-sampling, followed by \emph{Lloyd's algorithm}. Lloyd's algorithm 
is a local search heuristic. Starting with initial centers, it alternatingly assigns points to their closest center and computes the optimum center for each cluster (the centroid). The two steps are repeated until a stopping criterion is met, for example until the algorithm is converged or has reached a maximum number of iterations. When run to convergence, the initial solution is refined to a local optimum. Because of the clever sampling, $k$-means++ is a $\mathcal{O}(\log k)$-approximation on expectation, and due to the combination with the refinement by Lloyd's algorithm, it has become the state-of-the-art algorithm for $k$-means in practive.

\texttt{CKLV-$k$-means++} computes a clustering on the fairlet representatives with $k$-means++ and then assigns both points in a fairlet to the center that their representative is assigned to. 

\begin{algorithm}
\caption{CKLV-$k$-means++}
\label{alg:fairkmeanspp:mainpart}
\let\oldnl\nl
\newcommand{\nonl}{\renewcommand{\nl}{\let\nl\oldnl}}
1: Compute fairlet representatives $F$ with Algorithm~\ref{alg:fairlets:mainpart} or \ref{alg:fairlets:weighted}\\
2: Run standard $k$-means++ on $F$ and assign fairlet points accordingly
\end{algorithm}
\vspace*{-0.5cm}
\begin{algorithm}
\caption{Reassigned-CKVL}
\label{alg:reassigned-ckvl:mainpart}
\let\oldnl\nl
\newcommand{\nonl}{\renewcommand{\nl}{\let\nl\oldnl}}
1: Compute a center set $C$ with Algorithm~\ref{alg:fairkmeanspp:mainpart}\\
2: Compute an optimal fair assignment of all points in $P$ to $C$\\
\end{algorithm}

Both variants compute a $\mathcal{O}(\log k)$-approximation to the fair $k$-means problem. 
The second variant stems from the fact that a solution computed by $k$-means++ on the fairlet representatives can be improved by improving the fair assignment of the points.

Finally, we also present a direct extension of Lloyd's algorithm to the fair $k$-means setting.
The idea is to remove the loss in the quality due to the fairlet approach. We use $k$-means++ on the fairlets to obtain an initial solution, but then run a fair variant of Lloyd's algorithm on the actual data instead of running Lloyd's algorithm on the fairlets.
When doing Lloyd's, it is the assignment step which becomes nontrivial. Changing the center of a cluster to its centroid does not violate the fairness constraint, and it is still the optimal choice for the cluster. The assignment step, however, does no longer produce legal clusterings. We have to replace it with the fair assignment step. Because of this, fair Lloyd's is not a competitive algorithm when applied to the whole input because it becomes very slow. On the other hand, it computes the best solutions in our experiments. As we will see, the use of coresets speeds up the computation so much that it allows us to use fair Lloyd's to improve the solution quality in reasonable time.

\begin{algorithm}
\caption{Fair $k$-means++}
\label{alg:fairkmeans:mainpart}
\let\oldnl\nl
\newcommand{\nonl}{\renewcommand{\nl}{\let\nl\oldnl}}
1: Compute initial centers $C_0$ by Algorithm~\ref{alg:fairkmeanspp:mainpart}\\
2: For all $i\ge 0$, unless a stopping criterion is met:\\
3: \quad Assign every point to a center $C_i$ by evoking Algorithm~\ref{alg:fairassignment}, partitioning $P$ into $P_i^1,\ldots,P_i^k$\\
4: \quad Set $C_{i+1}=\{ \mu(P_i^j) \mid j \in [k]\}$
\end{algorithm}

From a theretical point of view, we also state how to obtain a PTAS for fair $k$-means clustering by following known techniques (Algorithm~\ref{alg:k-meanstheory} in the Section~\ref{sec:appendix:algorithms}).
We also extend the PTAS to the streaming setting. For this we develop a novel combination of the coreset construction with a sketching technique due to Cohen et. al.~\cite{CohenEMMP15} which may be of independent interest. For sketching techniques, the input points are usually represented by a matrix $A \in \mathbb{R}^{n \times d}$. Our theoretical streaming algorithm can be found in the Section~\ref{sec:appendix:algorithms} (Algorithm~\ref{alg:movementsketch}). It reduces the dimensionality of the input points in a clever way. We obtain the following result.

\begin{restatable}{theorem}{thmreduction}
\label{thm:reduction1}
Let $0<\varepsilon < 1/2$.
Assume there is streaming algorithm $ALG$ that receives the rows of a matrix $A\in \mathbb{R}^{n \times d}$ and
maintains an $(\epsilon,k)$-coreset $T$ with the following property: We can replace weighted points in $T$
by a corresponding number of copies to obtain a matrix $A'$ such that 
$
\sum_{i=1}^n \|A_{i*}-A_{i'*}\|^2 \le \frac{\varepsilon^2}{16} \cdot \text{OPT}.
$
Furthermore, assume that $ALG$ uses $f(k,\varepsilon,d,\log n)$ space. If we use $ALG$ in Step $2$ of Algorithm
\ref{alg:movementsketch}, then Algorithm \ref{alg:movementsketch} will use $f(k,\varepsilon/25,c' \cdot (k/\varepsilon)^2,\log n)\cdot d + O(kd/\epsilon^2)$ space 
to compute a set of centers $C$ with 
$$
\costf(P,C) \le \gamma (1+\varepsilon) \cdot \text{OPT}
$$
where $\text{OPT}$ is the cost of an optimal solution for $A$ and $c'>0$ is a constant  such that the guarantees of Theorem 12 from
\cite{CohenEMMP15} are satisfied.
\end{restatable}

Notice that combining this algorithm with Corollary~\ref{cor:fairbico} replaces the exponential dependency on $d$ with an exponential dependency on $k$. It is thus viable for very small values of $k$.

\section{Empirical evaluation}

Here we evaluate the approximation algorithms and our coreset approach empirically.

\paragraph{Algorithms} We compare \emph{CKLV-$k$-means++} (Algorithm~\ref{alg:fairkmeanspp:mainpart}), \emph{Reassigned-{CKLV}} (Algorithm~\ref{alg:reassigned-ckvl:mainpart}) and \emph{fair $k$-means++} (Algorithm~\ref{alg:fairkmeans:mainpart}). When we need to execute the algorithms on coresets, which are weighted, we use the algorithms for weighted inputs in Section~\ref{sec:appendix:weighted}.

\paragraph{Data sets} We use the same data sets as studied in~\cite{CKLV17} and choose the same sensitive dimensions.
The data sets all originate from the UCI library~\cite{UCI}.
We processed the data sets, converted or deleted non-numerical features and deleted data points with missing entries, 
and made sure that the data sets are balanced. We get data sets with the following properties:
\texttt{Adult:} A US census record data set from 1994 ($n=21542; d=6$; sensitive attribute: gender).
\texttt{Diabetes}: A data set about a study with diabetis patients ($n=94116; d=29$; sensitive attribute: marital status)
\texttt{Bank:} A data set about a marketing campaign of a Portuguese bank ($n=34600; d=11$; sensitive attribute: gender).



\paragraph{Setting}
All implementations were made in C or C++. For the minimum cost flow algorithm, we used capacity scaling as implemented in the C++ \emph{Library for Efficient Modeling and Optimization in Networks (LEMON)}~\cite{lemon}. As the coreset algorithm, we use BICO~\cite{FGSSS13}, following Corollary~\ref{cor:fairbico}. We (heuristicaally) set the coreset size to $200k$.
%
%
All experiments were run on one core of a Intel(R) Xeon(R) E3-1240 v6 processor with 32GB of main memory. As the stopping criterion for all Lloyd's variants, we use a maximum of $100$ iterations.
Since the $k$-means++ seeding introduces (a small amount of) variance, we repeat all experiments five times to account for the randomness.

\begin{figure}
\begin{tikzpicture}[scale=0.9]
\begin{axis}[xlabel=Number of points $n$,ylabel=Cost,thick, legend pos=south east,title=data set \texttt{diabetic},legend style={cells={align=left}},ymin=0.87,ymax=1.06,]
\addplot table [x=n, y=input kmeans++, col sep=semicolon] {\resultspath/diabetic-balancedCOR.auswertung.normalized.csv};
\addplot table [x=n, y=coreset kmeans++ c, col sep=semicolon] {\resultspath/diabetic-balancedCOR.auswertung.normalized.csv};
\legend{\texttt{Fair $k$-means++} on input,\texttt{Fair $k$-means++} on coreset}
\end{axis}
\end{tikzpicture}
\begin{tikzpicture}[scale=0.9]
\begin{axis}[xlabel=Number of points $n$,ylabel=Running time,thick, legend pos=north west,title=data set \texttt{diabetic},legend style={cells={align=left}}]
\addplot table [x=n, y expr=\thisrowno{1}+\thisrowno{2}+\thisrowno{3}, col sep=semicolon] {\resultspath/diabetic-balancedCOR.auswertung.time.csv};
\addplot table [x=n, y expr=\thisrowno{4}+\thisrowno{5}+\thisrowno{6}+\thisrowno{8}, col sep=semicolon] {\resultspath/diabetic-balancedCOR.auswertung.time.csv};
\addplot table [x=n, y expr=\thisrowno{1}+\thisrowno{2},col sep=semicolon]{\resultspath/diabetic-balancedCOR.fairalgo.auswertung.time.csv};
\legend{\texttt{Fair $k$-means++} on input,\texttt{Fair $k$-means++} on coreset,\texttt{CKLV-$k$-means++} on input}
\end{axis}
\end{tikzpicture}
\begin{tikzpicture}[scale=0.9]
\begin{axis}[xlabel=Number of points $n$,ylabel=Cost,thick, legend pos=south east,title=data set \texttt{bank},legend style={cells={align=left}},ymin=0.87,ymax=1.06,]
\addplot table [x=n, y=input kmeans++, col sep=semicolon] {\resultspath/bank-full-balancedCOR.auswertung.normalized.csv};
\addplot table [x=n, y=coreset kmeans++ c, col sep=semicolon] {\resultspath/bank-full-balancedCOR.auswertung.normalized.csv};
\legend{\texttt{Fair $k$-means++} on input,\texttt{Fair $k$-means++} on coreset}
\end{axis}
\end{tikzpicture}
\begin{tikzpicture}[scale=0.9]
\begin{axis}[xlabel=Number of points $n$,ylabel=Running time,thick, legend pos=north west,title=data set \texttt{bank},legend style={cells={align=left}}]
\addplot table [x=n, y expr=\thisrowno{1}+\thisrowno{2}+\thisrowno{3}, col sep=semicolon] {\resultspath/bank-full-balancedCOR.auswertung.time.csv};
\addplot table [x=n, y expr=\thisrowno{4}+\thisrowno{5}+\thisrowno{6}+\thisrowno{8}, col sep=semicolon] {\resultspath/bank-full-balancedCOR.auswertung.time.csv};
\addplot table [x=n, y expr=\thisrowno{1}+\thisrowno{2},col sep=semicolon]{\resultspath/bank-full-balancedCOR.fairalgo.auswertung.time.csv};
\legend{\texttt{Fair $k$-means++} on input,\texttt{Fair $k$-means++} on coreset,\texttt{CKLV-$k$-means++} on input}
\end{axis}
\end{tikzpicture}
\begin{tikzpicture}[scale=0.9]
\begin{axis}[xlabel=Number of points $n$,ylabel=Cost,thick, legend pos=south east,title=data set \texttt{adult},legend style={cells={align=left}},ymin=0.87,ymax=1.06,]
\addplot table [x=n, y=input kmeans++, col sep=semicolon] {\resultspath/adult-balancedCOR.auswertung.normalized.csv};
\addplot table [x=n, y=coreset kmeans++ c, col sep=semicolon] {\resultspath/adult-balancedCOR.auswertung.normalized.csv};
\legend{\texttt{Fair $k$-means++} on input,\texttt{Fair $k$-means++} on coreset}
\end{axis}
\end{tikzpicture}
\begin{tikzpicture}[scale=0.9]
\begin{axis}[xlabel=Number of points $n$,ylabel=Running time,thick, legend pos=north west,title=data set \texttt{adult},legend style={cells={align=left}},skip coords between index={10}{12}]
\addplot table [x=n, y expr=\thisrowno{1}+\thisrowno{2}+\thisrowno{3}, col sep=semicolon] {\resultspath/adult-balancedCOR.auswertung.time.csv};
\addplot table [x=n, y expr=\thisrowno{4}+\thisrowno{5}+\thisrowno{6}+\thisrowno{8}, col sep=semicolon] {\resultspath/adult-balancedCOR.auswertung.time.csv};
\addplot table [x=n, y expr=\thisrowno{1}+\thisrowno{2},col sep=semicolon]{\resultspath/adult-balancedCOR.fairalgo.auswertung.time.csv};
\legend{\texttt{Fair $k$-means++} on input,\texttt{Fair $k$-means++} on coreset,\texttt{CKLV-$k$-means++} on input}
\end{axis}
\end{tikzpicture}
\caption{Evaluation of the coreset with respect to quality decrease and runtime improvement. The left side shows the cost of the computed solutions. The highest deviation occurs for \texttt{bank}. The right side shows the improvement in the running time.\label{fig:coreseteval}}
\end{figure}
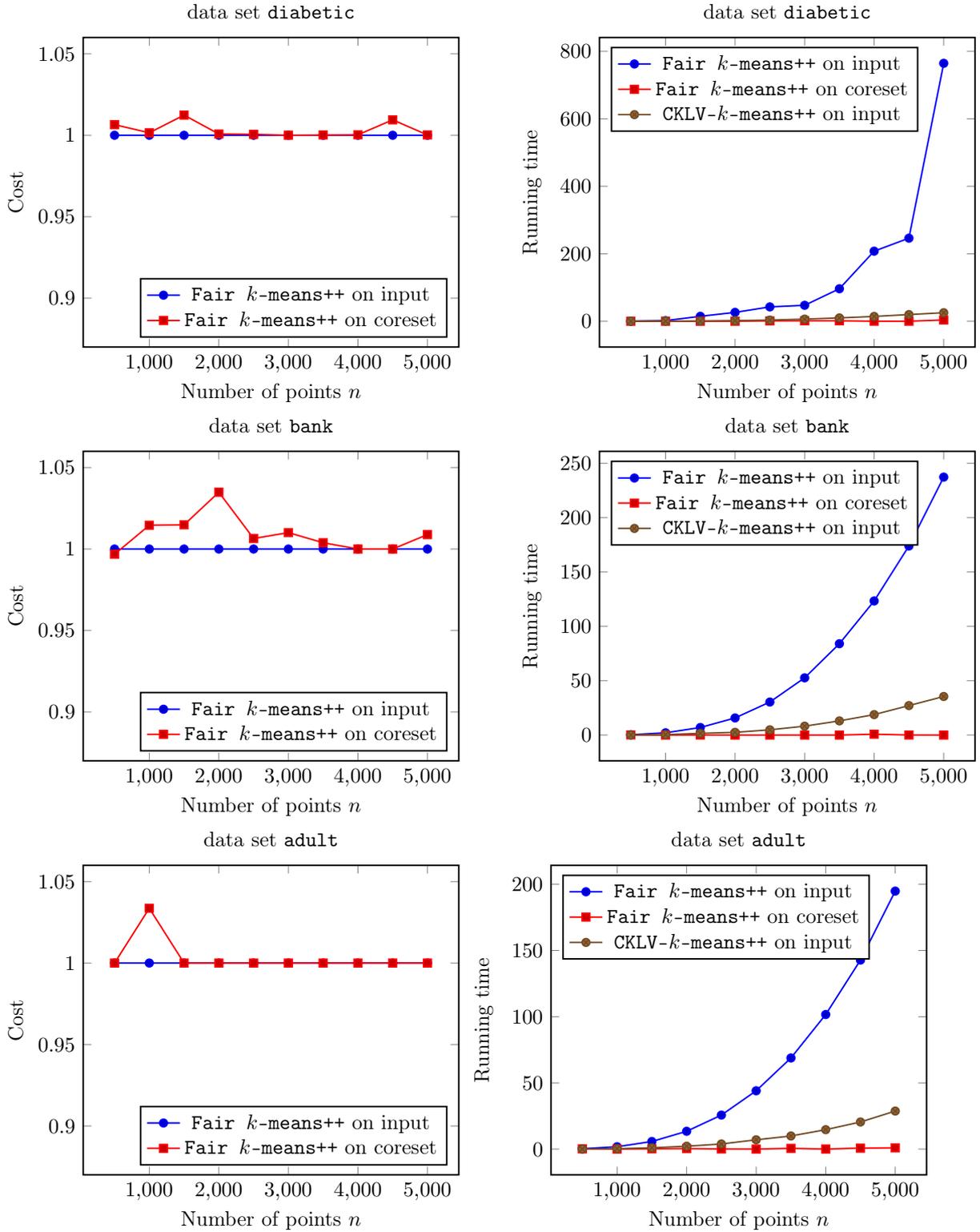

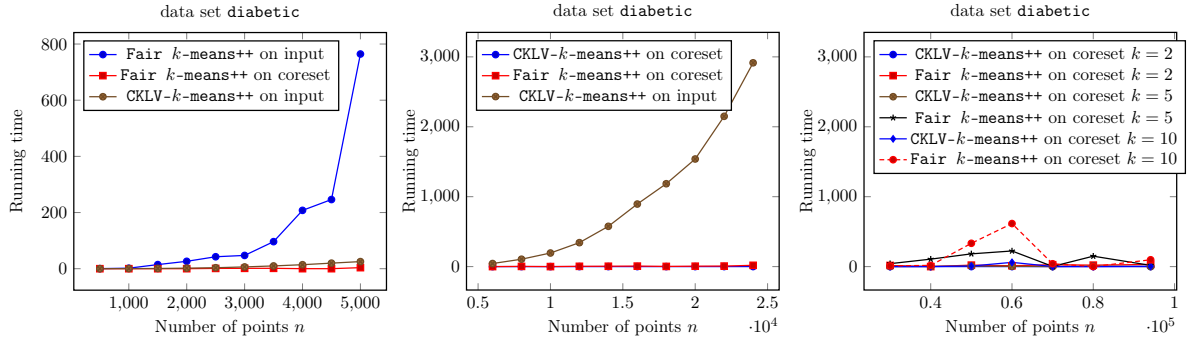
\begin{figure}
\begin{tikzpicture}[scale=0.6]
\begin{axis}[xlabel=Number of points $n$,ylabel=Running time,thick, legend pos=north west,title=data set \texttt{diabetic},legend style={cells={align=left}}]
\addplot table [x=n, y expr=\thisrowno{1}+\thisrowno{2}+\thisrowno{3}, col sep=semicolon] {\resultspath/diabetic-balancedCOR.auswertung.time.csv};
\addplot table [x=n, y expr=\thisrowno{4}+\thisrowno{5}+\thisrowno{6}+\thisrowno{8}, col sep=semicolon] {\resultspath/diabetic-balancedCOR.auswertung.time.csv};
\addplot table [x=n, y expr=\thisrowno{1}+\thisrowno{2},col sep=semicolon]{\resultspath/diabetic-balancedCOR.fairalgo.auswertung.time.csv};
\legend{\texttt{Fair $k$-means++} on input,\texttt{Fair $k$-means++} on coreset,\texttt{CKLV-$k$-means++} on input}
\end{axis}
\end{tikzpicture}
\begin{tikzpicture}[scale=0.6]
\begin{axis}[xlabel=Number of points $n$,ylabel=Running time,thick, legend pos=north west,title=data set \texttt{diabetic},legend style={cells={align=left}},ymax=3350]
\addplot table [x=n, y expr=\thisrowno{1}+\thisrowno{2}+\thisrowno{3},col sep=semicolon]{\resultspath/diabetic.lastexperiment.auswertung.time.csv};
\addplot table [x=n, y expr=\thisrowno{1}+\thisrowno{4}+\thisrowno{5},col sep=semicolon]{\resultspath/diabetic.lastexperiment.auswertung.time.csv};
\addplot table [x=n, y expr=\thisrowno{6}+\thisrowno{7},col sep=semicolon]{\resultspath/diabetic.lastexperiment.auswertung.time.csv};
\legend{\texttt{CKLV-$k$-means++} on coreset,\texttt{Fair $k$-means++} on coreset,\texttt{CKLV-$k$-means++} on input}
\end{axis}
\end{tikzpicture}
\begin{tikzpicture}[scale=0.6]
\begin{axis}[xlabel=Number of points $n$,ylabel=Running time,thick, legend pos=north west,title=data set \texttt{diabetic},legend style={cells={align=left}},ymax=3350]
\addplot table [x=n, y expr=\thisrowno{2}+\thisrowno{3}+\thisrowno{4},col sep=semicolon]{\resultspath/diabetic.lastexperimentBIG.auswertung.time400.csv};
\addplot table [x=n, y expr=\thisrowno{2}+\thisrowno{5}+\thisrowno{6},col sep=semicolon]{\resultspath/diabetic.lastexperimentBIG.auswertung.time400.csv};
\addplot table [x=n, y expr=\thisrowno{2}+\thisrowno{3}+\thisrowno{4},col sep=semicolon]{\resultspath/diabetic.lastexperimentBIG.auswertung.time1000.csv};
\addplot table [x=n, y expr=\thisrowno{2}+\thisrowno{5}+\thisrowno{6},col sep=semicolon]{\resultspath/diabetic.lastexperimentBIG.auswertung.time1000.csv};
\addplot table [x=n, y expr=\thisrowno{2}+\thisrowno{3}+\thisrowno{4},col sep=semicolon]{\resultspath/diabetic.lastexperimentBIG.auswertung.time2000.csv};
\addplot table [x=n, y expr=\thisrowno{2}+\thisrowno{5}+\thisrowno{6},col sep=semicolon]{\resultspath/diabetic.lastexperimentBIG.auswertung.time2000.csv};
\legend{\texttt{CKLV-$k$-means++} on coreset $k=2$,\texttt{Fair $k$-means++} on coreset $k=2$, \texttt{CKLV-$k$-means++} on coreset $k=5$,\texttt{Fair $k$-means++} on coreset $k=5$,\texttt{CKLV-$k$-means++} on coreset $k=10$,\texttt{Fair $k$-means++} on coreset $k=10$}
\end{axis}
\end{tikzpicture}
\caption{Runtime development for \texttt{diabetic}\label{fig:runtime:diabetic:mainpart}.}
\end{figure}

\paragraph{Results}
We did experiments with small, medium and large data sets.
To evaluate the effect of the coreset on the quality, we can only use small data sets because we need to compute solutions on the input data as well to compare the quality. Thus, we use small subsampled data sets of small increasing size between $500$ and $5000$. 
We then compute a $2$-clustering on the data (the small choice of $k$ is due to the high running time of \texttt{Fair $k$-means++}) and compare the quality of the solutions computed on the input set with those computed on the coreset (evaluated on the input set). 
The left side of Figure~\ref{fig:coreseteval} shows the coreset quality for all data sets. The error is low, the largest error (up to 3,5\%) occurs for the data set \texttt{bank}.

On the right side Figure~\ref{fig:coreseteval}, we depict the running time. We can clearly see that the runtime of \texttt{Fair $k$-means++} does not scale even to medium sized data. The coreset approach is relatively unaffected by the increasing data size, since the coreset computation is nearly linear and the coreset size only depends on $k$. \texttt{CKLV-$k$-means++} is a much faster alternative to \texttt{Fair $k$-means++}, but we already anticipate that it will not scale to big data.

To further evaluate the running times on larger data sets, we did further experiments for the data set \texttt{diabetic} (because it is the largest of the data sets). The first diagram of Figure~\ref{fig:runtime:diabetic:mainpart} shows the runtime of all algorithms on the small data sets (identical to the corresponding diagram in Figure~\ref{fig:coreseteval}). 
In the second diagram, we see results for data sets of medium size. As expected, \texttt{CKLV-$k$-means++} can handle middle sized data set. However, at $n=26.000$, the memory requirement of the fairlet computation became larger than the main memory, making the approach basically infeasible. 
Notice that we could speed up the algorithm somewhat by improving the implementation. However, there is no simple way of bypassing the problem that the bipartite graph needed for the fair assignments grows quadratically in the input point size. 
The computations on the coreset, however, are unaffected by the increase in the data. We can easily perform even \texttt{Fair $k$-means++} on the coreset.

Indeed, the running time for the coreset algorithms mainly scales with $k$: That is because the coreset size depends on $k$, so the computation time on the coreset increases with $k$. We see this on the right side of Figure~\ref{fig:runtime:diabetic:mainpart}, where we see the running times of the algorithms on the \texttt{diabetic} data set up to its full size, for $k=2$, $k=5$ and for $k=10$. As we can see, the choice of $k$ has more effect on the running time than increasing the size of the input data.

Figure~\ref{fig:algocomparison} compares the solution quality of \texttt{Fair $k$-means++}, \texttt{CKLV-$k$-means++} and \texttt{Reassigned-CKLV} on subsampled data sets of size 1000 with increasing values of $k$. The diagrams are normalized to the cost of  \texttt{CKLV-$k$-means++} (blue).
 
The left side shows the quality on all data sets.  The right side shows the costs of the fairlets for all data sets. The cost of the fairlet micro-clustering can be seen as the \lq cost of fairness\rq. It is a lower bound on the cost of any clustering.

As we can see, the fairlets are very expensive for \texttt{bank} and \texttt{adult} (probably due to a large gender bias that is present in these data sets). Compared to the cost for the fairlets, the actual clustering cost gets negligible with increasing $k$. This reflects in the quality experiments as well. For \texttt{bank} and \texttt{adult}, the algorithms perform basically identically well. 

For data set \texttt{diabetic}, the price of fairness is not as high, leaving some room for optimization to the algorithms. Here we see that \texttt{Fair $k$-means++} achieves the best cost. 
\texttt{Fair $k$-means++} (red) gains up to 5\% in quality on \texttt{diabetic} compared to the optimized fairlet approach \texttt{Reassigned-CKLV}. Compared to the unoptimized solution of \texttt{CKLV-$k$-means++}, the gain is higher.

\begin{figure}
\begin{tikzpicture}[scale=0.9]
\begin{axis}[xlabel=Number of centers $k$,ylabel=Cost / Fairlet+$k$-means++ cost,thick, legend pos=south east,ymin=0.87,ymax=1.05,skip coords between index={0}{2},title=data set \texttt{diabetic}]
\addplot table [x=k, y=fairlet algo cost, col sep=semicolon] {\resultspath/diabetic-balanced1000.auswertung.normalized.csv};
\addplot table [x=k, y=Lloyd with Fairlet Init, col sep=semicolon] {\resultspath/diabetic-balanced1000.auswertung.normalized.csv};
\addplot table [x=k, y=sum, col sep=semicolon] {\resultspath/diabetic-balanced1000.auswertung.normalized.csv};
\legend{\texttt{Reassigned-CKLV}, \texttt{Fair $k$-means++}, \texttt{CKLV-$k$-means++}}
\end{axis}
\end{tikzpicture}
\begin{tikzpicture}[scale=0.9]
\begin{axis}[xlabel=Number of centers $k$,ylabel=Cost,thick, legend pos=north east,skip coords between index={0}{2},title=data set \texttt{diabetic},legend style={cells={align=left}}] 
\addplot table [x=k, y=fairlet algo cost, col sep=semicolon] {\resultspath/diabetic-balanced1000.auswertung.csv};
\addplot table [x=k, y=fairlet cost, col sep=semicolon,mark=none] {\resultspath/diabetic-balanced1000.auswertung.csv};
\legend{\texttt{Reassigned-CKLV}, Fairlet Cost}
\end{axis}
\end{tikzpicture}
\begin{tikzpicture}[scale=0.9]
\begin{axis}[xlabel=Number of centers $k$,ylabel=Cost / Fairlet+$k$-means++ cost,thick, legend pos=south east,ymin=0.87,ymax=1.05,skip coords between index={0}{2},title=data set \texttt{bank}]
\addplot table [x=k, y=fairlet algo cost, col sep=semicolon] {\resultspath/bank-full-balanced1000.auswertung.normalized.csv};
\addplot table [x=k, y=Lloyd with Fairlet Init, col sep=semicolon] {\resultspath/bank-full-balanced1000.auswertung.normalized.csv};
\addplot table [x=k, y=sum, col sep=semicolon] {\resultspath/bank-full-balanced1000.auswertung.normalized.csv};
\legend{\texttt{Reassigned-CKLV}, \texttt{Fair $k$-means++}, \texttt{CKLV-$k$-means++}}
\end{axis}
\end{tikzpicture}
\begin{tikzpicture}[scale=0.9]
\begin{axis}[xlabel=Number of centers $k$,ylabel=Cost,thick, legend pos=north east,skip coords between index={0}{2},title=data set \texttt{bank},legend style={cells={align=left}}]
\addplot table [x=k, y=fairlet algo cost, col sep=semicolon] {\resultspath/bank-full-balanced1000.auswertung.csv};
\addplot table [x=k, y=fairlet cost, col sep=semicolon,mark=none] {\resultspath/bank-full-balanced1000.auswertung.csv};
\legend{\texttt{Reassigned-CKLV}, Fairlet Cost}
\end{axis}
\end{tikzpicture}
\begin{tikzpicture}[scale=0.9]
\begin{axis}[xlabel=Number of centers $k$,ylabel=Cost / Fairlet+$k$-means++ cost,thick, legend pos=south east,ymin=0.87,ymax=1.05,skip coords between index={0}{2},title=data set \texttt{adult}]
\addplot table [x=k, y=fairlet algo cost, col sep=semicolon] {\resultspath/adult-balanced1000.auswertung.normalized.csv};
\addplot table [x=k, y=Lloyd with Fairlet Init, col sep=semicolon] {\resultspath/adult-balanced1000.auswertung.normalized.csv};
\addplot table [x=k, y=sum, col sep=semicolon] {\resultspath/adult-balanced1000.auswertung.normalized.csv};
\legend{\texttt{Reassigned-CKLV}, \texttt{Fair $k$-means++}, \texttt{CKLV-$k$-means++}}
\end{axis}
\end{tikzpicture}
\begin{tikzpicture}[scale=0.9]
\begin{axis}[xlabel=Number of centers $k$,ylabel=Cost,thick, legend pos=north east,skip coords between index={0}{2},title=data set \texttt{adult},legend style={cells={align=left}}]
\addplot table [x=k, y=fairlet algo cost, col sep=semicolon] {\resultspath/adult-balanced1000.auswertung.csv};
\addplot table [x=k, y=fairlet cost, col sep=semicolon,mark=none] {\resultspath/adult-balanced1000.auswertung.csv};
\legend{\texttt{Reassigned-CKLV}, Fairlet Cost}
\end{axis}
\end{tikzpicture}
\caption{Comparison of the solution quality. The left side compares the three approaches. The right side shows how the clustering cost relates to the cost of the fairlet computation. For \texttt{bank} and \texttt{adult}, the clustering cost is dominated by the \lq cost of fairness\rq. For \texttt{diabetic}, this is not the case, implying that there is optimization potential for the approximation algorithms.\label{fig:algocomparison}}
\end{figure}
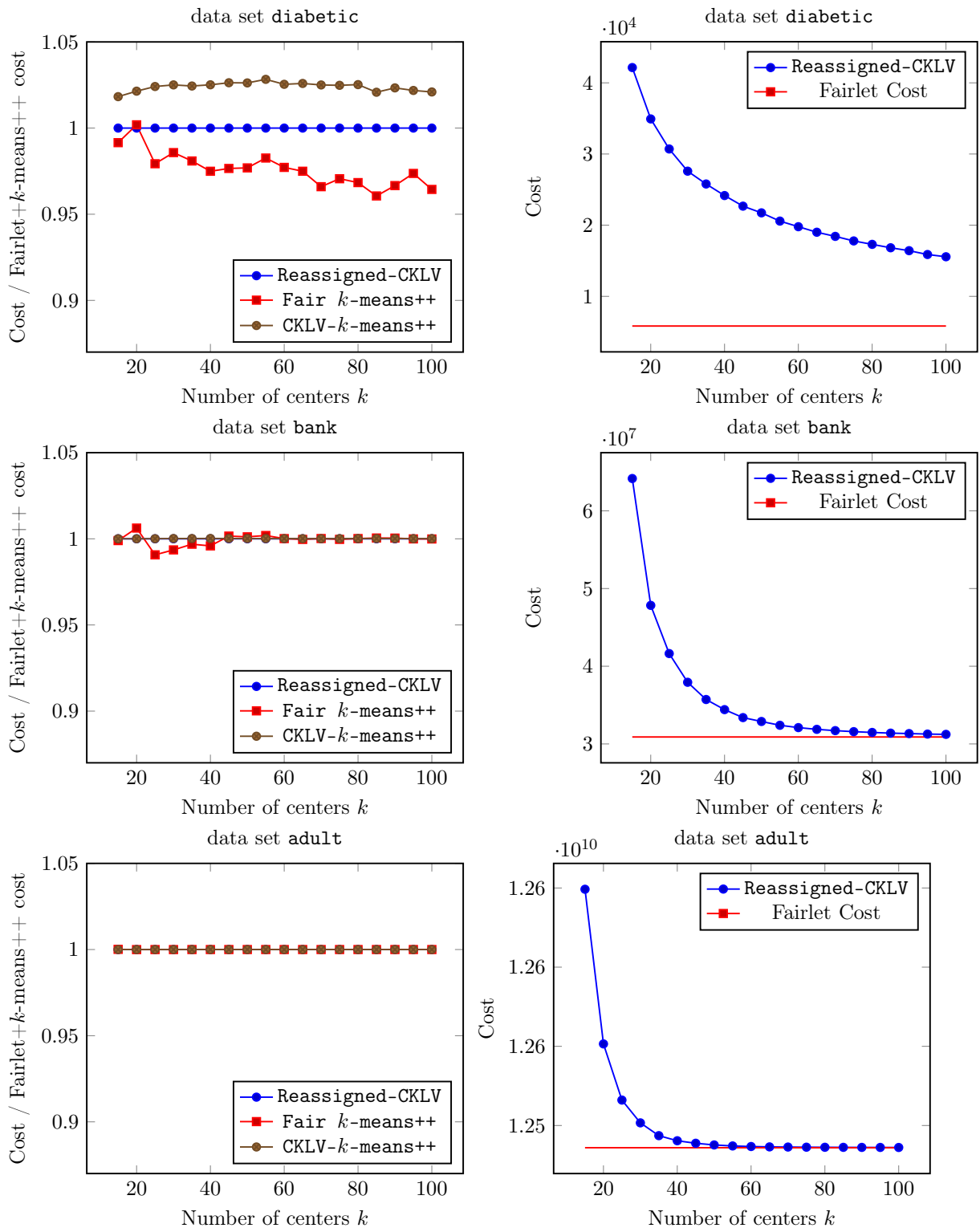

\section{Approximation algorithms for fair \texorpdfstring{$k$}{k}-means (full)}\label{sec:appendix:algorithms}
\subsection{Lloyd's algorithm and \texorpdfstring{$k$}{k}-means++}\label{subsec:appendix:lloydkmeans++}

The most known heuristic for (colorless) $k$-means in practice is the $k$-means algorithm / Lloyd's algorithm~\cite{L57}.
It is a local search heuristic that alternates an assignment step with recomputing the centers; it is of similar nature as EM type algorithms. We state the $k$-means algorithm as Algorithm~\ref{alg:kmeans}.

\begin{algorithm}
\caption{$k$-means/Lloyd's algorithm}
\label{alg:kmeans}
\let\oldnl\nl
\newcommand{\nonl}{\renewcommand{\nl}{\let\nl\oldnl}}
1: Compute initial centers $C_0$\\
2: For all $i\ge 0$, unless a stopping criterion is met:\\
3: \quad Assign every point to its closest center in $C_i$, partitioning $P$ into $P_i^1,\ldots,P_i^k$\\
4: \quad Set $C_{i+1}=\{ \mu(P_i^j) \mid j \in [k]\}$
\end{algorithm}

Both steps of the algorithm are improving (or at least do not worsen the solution): The cheapest way to assign points to centers is to assign every point optimally, and the centroid is always the best center for a cluster by Proposition~\ref{prop:zauberformel}.
A typical stopping criterion is when a given number of iterations is reached, or when the decrease of the objective function per iteration falls below a given threshold (notice, however, that this stopping criterion has to be parametrized in a data dependent fashion). When implementing Lloyd's, some care should be given to the possible event that a cluster runs out of points. A typical reaction to this is to choose a new center randomly (or by $k$-means++ seeding). 

The $k$-means algorithm crucially depends on the initialization. The de facto standard initialization method has become \emph{$k$-means++ seeding}, which guarantees a $\mathcal{O}(\log k)$ approximation; the combination of seeding and $k$-means++ is called the $k$-means++ algorithm, see Algorithm~\ref{alg:kmeanspp}.
\begin{algorithm}
\caption{$k$-means++}
\label{alg:kmeanspp}
\let\oldnl\nl
\newcommand{\nonl}{\renewcommand{\nl}{\let\nl\oldnl}}
1: Initialize $C=\{p\}$ with a random point $p\in P$\\
2: {\bf While} $|C|<k$\\
3: ~~~ add $q\in P$ to $C$ with probability $\frac{\underset{c\in C}{\min}\|q-c\|^2}{\sum_{q'\in P}\underset{c\in C}{\min}\|q-c\|^2}$\\
4: Run Algorithm~\ref{alg:kmeanspp} with $C_0=C$ as initial centers
\end{algorithm}

\subsection{Fairlets for \texorpdfstring{$k$}{k}-means}\label{subsec:appendix:fairletskmeans}
Chierichetti et. al. \cite{CKLV17} outline how to transfer approximation algorithms for clustering to the setting of fair clustering via the computation of \emph{fairlets}, but derive the algorithms only for $k$-center and $k$-median. 
We give a method to compute fairlets for $k$-means in Algorithm ~\ref{alg:fairlets}. Computing an $\alpha$-approximation for $k$-means on the fairlets provides a $(5\alpha+1)$-approximation for the fair $k$-means clustering problem:

\begin{algorithm}
\caption{Fairlet computation}
\label{alg:fairlets}
\let\oldnl\nl
\newcommand{\nonl}{\renewcommand{\nl}{\let\nl\oldnl}}
1: Let $B$ be the blue points and $R$ be the red points\\
2: For any $b \in B$, $r \in R$, set $c(r,b) = ||r-b||^2/2$\\
3: Consider the complete bipartite graph $G$ on $B$ and $R$\\
4: Compute a min cost perfect matching $M$ on $G$\\
5: For each edge $(r,b)\in M$, add $\mu(\{r,b\})$ to $F$\\
6: Output $F$
\end{algorithm}

\thmfairletskmeans*

\begin{proof}
Fix some optimal fair $k$-means clustering solution $C^\ast$ on $P$. 
Observe that $C^\ast$ partitions $P$ into $n/2$ pairs, each consisting of a blue and a red point. Let $(r,b)$ be a pair which is assigned to center $c \in C^\ast$ in the optimal solution. This costs more than assigning both $r$ and $b$ to their optimum center, namely, the centroid of $r$ and $b$. The cost of assigning $r$ and $b$ to $\mu(\{r\},\{b\})$ is $||r-b||^2/2$. Thus, the sum of $||r-b||^2/2$ summed over all pairs is at most $OPT$. 
In Algorithm~\ref{alg:fairlets}, we set the cost $c(r,b)$ to $||r-b||^2/2$. We just argued that there is a partitioning of the points into $n/2$ bichromatic pairs such that the sum of the costs is at most $OPT$. By computing a min cost perfect matching, we obtain the cheapest such partitioning, which thus also costs at most $OPT$.

Now let $C$ be any solution, and assign all points in $F$ to their closest center. 
By Proposition~\ref{prop:zauberformel}, 
\[
\sum_{\mu(\{r,b\})\in F} \min_{c \in C} ||r-c||^2+||b-c||^2 = 2\cost(F,C) + c(M).
\]
Thus, $2\cost(F,C)+c(M)$ is the cost that we pay if we insist on clustering both points of a fairlet together. If we compute an optimal assignment of $P$ to $C$, then the cost can only be smaller. 

It remains to argue that $2 \cost (F,C)+2c(M)$ is small.
First recall that 
\begin{equation}
\sum_{\mu(\{r,b\})\in F} ||r-b||^2 / 2 \le OPT\label{eqnaff},
\end{equation}
so $c(M) \le OPT$. 
Furthermore, say w.l.o.g. that $r$ is cheaper in $\cost(P,C)$ than $b$, and consider the center $c \in C$ which is the closest center for $r$. Then $r$ pays the same as it would pay in $\cost(P,C)$. We can estimate the cost of $b$ by $||b-c||^2 \le 2||b-r||^2 + 2||r-c||^2 \le 2||b-r||^2 + 0.5\cdot||r-c||^2 + 1.5\cdot||b-c'||^2$, where $c' \in C$ is the closest center to $b$. It we sum this up for all pairs, we get:
\begin{align*}
& 2 \cost (F,C)\\
\le & \sum_{\mu(\{r,b\})\in F} 2||b-r||^2 + 1.5\cdot||r-c||^2 + 1.5\cdot ||b-c'||^2\\
\le& 4 \cdot OPT + 1.5 \cdot \cost(P,C) \le 5.5 \cdot \costf(P,C),
\end{align*}
where we use \eqref{eqnaff} for the second inequality. 
If we now compute an $\alpha$-approximation on $F$, we get a solution with $2 \cost(F,C) \le 5.5 \alpha OPT$, leading to a total cost of $2 \cost(F,C) + c(M) \le (5.5\alpha +1) \cdot \costf(P,C)$.
\end{proof}


\subsection{Fair assignment}

\begin{algorithm}
\caption{Fair assignment}
\label{alg:fairassignment}
\let\oldnl\nl
\newcommand{\nonl}{\renewcommand{\nl}{\let\nl\oldnl}}
1: Let $B$ be blue points, $R$ red points and $C$ centers\\
2: $\forall b \in B, r \in R$, set $c(r,b) = \min_{c\in C}||r-c||^2+||b-c||^2$ and $cen(r,b)=\arg\min{c \in C} ||r-c||^2+||b-c||^2$\\ 
3: Consider the complete bipartite graph $G$ on $B$ and $R$\\
4: Compute a min cost perfect matching $M$ on $G$ wrt $c$\\
5: For each edge $(r,b)\in M$, assign $r$ and $b$ to $cen(r,b)$\\
6: Output $F$
\end{algorithm}

With the fairness constraint, assigning points to a given set of centers $C$ becomes non-trivial. Similarly to the fairlet computation, we can model the assignment step as a matching problem. The points are the vertices. Between any blue point $b$ and any red point $r$, we insert an edge whose cost is the minimum cost of assigning $b$ and $r$ to the same center. Now we want to match every blue to a red point while minimizing the cost, i.e., we look for a minimum cost perfect matching in a bipartite graph. This problem is polynomially solvable, e.g., by the Hungarian method. 

\subsection{Fairlets and fair assignment for weighted inputs}\label{sec:appendix:weighted}
\begin{algorithm}
\caption{Fairlet computation (weighted)}
\label{alg:fairlets:weighted}
\let\oldnl\nl
\newcommand{\nonl}{\renewcommand{\nl}{\let\nl\oldnl}}
1: Let $B$ be the blue points and $R$ be the red points\\
2: For any $b \in B$, $r \in R$, set $c(r,b) = ||r-b||^2/2$\\
3: Construct a complete bipartite network where every $b \in B$ is a source with supply $w(b)$,
every $r\in R$ is a sink with demand $w(b)$, every $b$ is connected to every $r$ with a directed edge of infinite capacity\\
4: Compute an integral minimum cost flow with respect to the edge costs $c$ that satisfies all supplies and demands\\
5: For all $r \in R,b \in B$, add $\mu(\{r,b\})$ to $F$ with weight $f((r,b))$\\
6: Output $F$
\end{algorithm}

\begin{algorithm}
\caption{Fair assignment (weighted)}
\label{alg:fairassignment:weighted}
\let\oldnl\nl
\newcommand{\nonl}{\renewcommand{\nl}{\let\nl\oldnl}}
1: Let $B$ be blue points, $R$ red points, $w$ the weights and $C$ the centers\\
2: $\forall b \in B, r \in R$, set $c(r,b) = \min_{c\in C}||r-c||^2+||b-c||^2$  and $cen(r,b)=\arg\min{c \in C} ||r-c||^2+||b-c||^2$\\
3: Construct a complete bipartite network where every $b \in B$ is a source with supply $w(b)$,
every $r\in R$ is a sink with demand $w(b)$, every $b$ is connected to every $r$ with a directed edge of infinite capacity\\
4: Compute an integral minimum cost flow with respect to the edge costs $c$ that satisfies all supplies and demands\\
5: For all $r \in R,b \in B$, assign $f((r,v))$ weight from $r$ and $b$ to $cen(r,b)$\\
6: Output $F$
\end{algorithm}Notice that we want to use our algorithms on coresets, i.e., on weighted inputs. Thus, we want to solve the assignment in a setting where the points have (splittable) weights, without incurring a runtime that depends on the vertex weights. In this context, it is more convenient to view the problem as a minimum cost flow problem. The blue points are the sources, the red points are the sinks, and the supplies and demands correspond to the point weights. The minimum cost flow problem can be solved in polynomial time and for integral capacities, supplies and demands, there is an optimal flow which is integral. A strongly polynomial algorithm computing an integral min cost flow is Enhanced Capacity Scaling with a running time of $\mathcal{O}((m \log n)(m+ n\log n))$. In our setting, $m \in \Theta(n^2)$, so the running time is $\mathcal{O}(n^4 \log n)$. The standard text book~\cite{AMO93} provides a detailed description of this and several other minimum cost flow algorithms. 

\begin{observation}
The fair assignment step can be performed optimally in strongly polynomial time, even if the points are weighted. The same is true for the fairlet computation.
\end{observation}

\subsection{\texttt{CKLV-\texorpdfstring{$k$}{k}-means++} and \texttt{Reassigned-CKVL}}
We can now combine the fairlet computation with $k$-means++. 
\begin{corollary}
Algorithm~\ref{alg:fairkmeanspp} computes an expected $\mathcal{O}(\log k)$-approximation for the fair $k$-means problem.
\end{corollary}

\begin{algorithm}
\caption{CKLV-$k$-means++}
\label{alg:fairkmeanspp}
\let\oldnl\nl
\newcommand{\nonl}{\renewcommand{\nl}{\let\nl\oldnl}}
1: Compute fairlet representatives $F$ with Algorithm~\ref{alg:fairlets} or \ref{alg:fairlets:weighted}\\
2: Run Algorithm~\ref{alg:kmeanspp} on $F$
\end{algorithm}

\begin{algorithm}
\caption{Reassigned-CKVL}
\label{alg:reassigned-ckvl}
\let\oldnl\nl
\newcommand{\nonl}{\renewcommand{\nl}{\let\nl\oldnl}}
1: Compute a center set $C$ with Algorithm~\ref{alg:fairkmeanspp}\\
2: Assign all points in $P$ to $C$ with Algorithm~\ref{alg:fairassignment} or \ref{alg:fairassignment:weighted}\\
\end{algorithm}

\subsection{Fair \texorpdfstring{$k$}{k}-means}

Finally, we want to adapt Lloyd's algorithm
 to the setting of fair $k$-means clustering. 
The idea behind this is to avoid the (constant-factor) loss of first computing fairlets and then using an approximation algorithm. 
As initialization, we use the $k$-means++ seeding on the fairlets, which guarantees that the outcome is a $\mathcal{O}(\log k)$ approximation on expectation. For the seeding, we do not perform Lloyd on the fairlets, but only the seeding part of $k$-means++.

When computing a fair clustering, it is the assignment step of Lloyd's algorithm which becomes nontrivial. Changing the center of a cluster to its centroid does not violate the fairness constraint, and it is still the optimal choice for the cluster. The assignment step, however, does no longer produce legal clusterings. We have to replace it with the fair assignment step. 

\begin{algorithm}
\caption{Fair $k$-means++}
\label{alg:fairkmeans}
\let\oldnl\nl
\newcommand{\nonl}{\renewcommand{\nl}{\let\nl\oldnl}}
1: Compute initial centers $C_0$ by Algorithm~\ref{alg:fairkmeanspp}\\
2: For all $i\ge 0$, unless a stopping criterion is met:\\
3: \quad Assign every point to a center $C_i$ by evoking Algorithm~\ref{alg:fairassignment} or \ref{alg:fairassignment:weighted}, partitioning $P$ into $P_i^1,\ldots,P_i^k$\\
4: \quad Set $C_{i+1}=\{ \mu(P_i^j) \mid j \in [k]\}$
\end{algorithm}

\subsection{A PTAS for fair \texorpdfstring{$k$}-means}\label{sec:ptas}

We next give an algorithm to efficiently compute a $(1+\epsilon)$-approximation. We remark that the running time of this
algorithm is worse than that of \cite{BJK18,DX15}. However, it can be easily adapted to work with weighted inputs. While we believe that in principle adapting the algorithms in \cite{BJK18,DX15} to the weighted case 
is possible, we preferred to stick to the simpler slightly worse result to keep the paper concise. 

\begin{restatable}{theorem}{inabatheorem}
Let $P\subseteq \REAL^d$ be a weighted point set of $n$ points such that half of the point weight is red and the other half is blue. Then we can compute a $(1+\epsilon)$-approximations to the fair $k$-means problem in time $n^{O(k/\epsilon)}$.
\end{restatable}

We will use the well-known fact that every cluster has a subset of $O(1/\epsilon)$ points such that their centroid is a $(1+\epsilon)$-approximation to the centroid of the cluster. We use the following lemma 
by Inaba et al. 

\begin{lemma}\cite{IKI94}
Let $P \subseteq \REAL^d$ be a set of points and let $S$ be a subset of $m$ points drawn independently and
uniformly at random from $P$. Let $c(P)= \frac{1}{|P|} \sum_{p\in P} p$ and  $c(S)= \frac{1}{|S|} \sum_{p\in S} p$ be the centroids of $P$ and $S$.
Then with probability at least $1-\delta$ we have
$$
\| \sum_{p\in P} \|p-c(S)\|_2^2 \le (1+\frac{1}{\delta m}) \cdot \|\sum_{p\in P} \|p-c(P)\|_2^2
\enspace.
$$
\end{lemma}

It immediately follows that for $m=\lceil 2/\epsilon \rceil$ there exists a subset $S$ of $m$ points that
satisfies the above inequality. The result can immediately be extended to the weighted case. This implies the following algorithm.

\begin{algorithm}
\caption{PTAS for fair $k$-means++}
\label{alg:k-meanstheory}
\let\oldnl\nl
\newcommand{\nonl}{\renewcommand{\nl}{\let\nl\oldnl}}
{\bf Input:} (Weighted) point set $P \subseteq \REAL^d$ \\
1: Consider all subsets $S \subseteq P$ of size $k \cdot \lceil 2/\epsilon \rceil$.\\
2: Partition $S$ into $k$ sets $C_1,\dots, C_k$ of size $\lceil 2/\epsilon \rceil$.\\
3: Solve the fair assignment problem for $P$ and $c(C_1), \dots, c(C_k)$ (with Algorithm~\ref{alg:fairassignment} or \ref{alg:fairlets:weighted},respectively)\\
4: Return the best solution computed above\\
\end{algorithm}

The running time of the algorithm is $n^{O(k/\epsilon)}$ since line two can be implemented in 
$k^{O(k/\epsilon)}$ time and the partition problem can be solved in $n^{O(1)}$ time. This implies the theorem.

\subsection{A Streaming PTAS for Small \texorpdfstring{$k$}{k}}
\label{sec:dimred}

We would like to extend the PTAS to the streaming setting, using our coreset. Applying Corollary~\ref{cor:fairbico} directly incurs an exponential dependency on the dimension $d$. 
The standard way to avoid this is to project the data onto the first $k/\varepsilon$ principal components, see~\cite{CohenEMMP15,FeldmanSS13}, and then to use a technique called merge-and-reduce. 
Unfortunately, merge-and-reduce technique requires a rescaling of $\varepsilon$ by a factor of $\log n$. In other words, the resulting streaming coreset will have a size $\exp(\left(\frac{\log n}{\varepsilon}\right),{k\cdot \frac{\log n}{\varepsilon}})$, which is even larger than the input size. 
To avoid this, we show how to make use of oblivious random projections to reduce the dependency of the dimension for movement-based coreset constructions, and also recover a $(1+\varepsilon)$ approximate solution.

We review some of the algebraic properties.
Given a matrix $A\in \mathbb{R}^{n\times d}$, we define the Frobenius norm as $\|A\|_F = \sqrt{\sum_{i=1}^n \|A_{i*}\|^2}$, where $A_{i*}$ is the $i$th row of $A$.
For $k$-means, we will consider the rows of $A$ to be our input points.
The spectral norm $\|A\|_2$ is the largest singular value of $A$. 

Let us now consider the $n$-vector $x=\mathbf{1}\cdot \frac{1}{\sqrt{n}}$. $x$ is a unit vector, i.e. $\|x\|_2 = 1$, and moreover, due to Proposition~\ref{prop:zauberformel}, the rows of $xx^TA$ are $\mu(A)^T$. Hence $\|A-xx^TA\|_F^2$ is the optimal $1$-means cost of $A$. This may be extended to an arbitrary  number of centers by considering the $n$ by $k$ clustering matrix $X$ with $X_{i,j}=\begin{cases} \sqrt{1/|C_j|} & \text{if } A_{i*}\in \text{ cluster }C_j\\
0 &\text{otherwise}\end{cases}.$
$XX^T$ is an orthogonal projection matrix and $\|A-XX^TA\|_F^2$ corresponds to the $k$-means cost of the clusters $C_1,\ldots ,C_k$.
If we lift the clustering constraints on $X$ and merely assume $X$ to be orthogonal and rank $k$,  $\|A-XX^TA\|_F^2$ becomes the rank $k$ approximation problem. The connection between rank $k$ approximation and $k$-means is well established, see for 
example~\cite{BZMD15,DrineasFKVV04,FeldmanSS13}. Specifically, we aim for the following guarantee.
\begin{definition}[Definition 1 of~\cite{CohenEMMP15}]
\label{def:sketch}
$\tilde{A}\in \mathbb{R}^{n\times d'}$ is a rank $k$-projection-cost preserving sketch of $A\in \mathbb{R}^{n\times d}$ with error $0<\varepsilon<1$ if, for all rank $k$ orthogonal projection matrices $XX^T\in\mathbb{R}^{n\times n}$,
\[\|\tilde{A}-XX^T\tilde{A}\|_F^2 + c \in (1\pm\varepsilon)\cdot \|A-XX^TA\|_F^2,\] 
for some fixed non-negative constant $c$ that may depend on $A$ and $\tilde{A}$, but is independent of $XX^T$.
\end{definition}
Our choice of $\tilde{A}$ is $AS$, where $S$ is a scaled Rademacher matrix of target dimension $m\in O(k/\varepsilon^2)$, see Theorem 12 of~\cite{CohenEMMP15}. 
In this case $c=0$. 
\begin{algorithm}
\caption{Dimension-Efficient Coreset Streaming}
\label{alg:movementsketch}
\let\oldnl\nl
\newcommand{\nonl}{\renewcommand{\nl}{\let\nl\oldnl}}
{\bf Input:} Point set $A$ processed in a stream\\
1: Initialize $S\in\mathbb{R}^{d\times m}$\\
2: Maintain a movement-based coreset $T$ of $AS$ \\ 
3: Let $\pi^{-1}(T_{i*})$ be the set of rows of $AS$ that are moved to $T_{i*}$\\
4: Let $\pi^{-1}_A(T_{i*})$ be the set of corresponding rows of $A$\\
5: Maintain $|\pi^{-1}_A(T_{i*})|$ and the linear sum $L(T_{i*})$ of the rows in $\pi^{-1}_A(T_{i*})$\\
6: Solve $(\alpha,\beta)$-fair $k$-means on $T$ using a $\gamma$-approximation algorithm $\leadsto$ clustering $C_1,\dots,C_k$ 
\\
7: For each cluster $C_j$ return the center $\frac{1}{\sum_{T_{i*} \in C_j} |\pi^{-1}_A(T_{i*})|} \cdot \sum_{T_{i*} \in C_j} L(T_{i*})$ \\
\end{algorithm}

We combine oblivious sketches with movement-based coreset constructions in Algorithm~\ref{alg:movementsketch}.
The general idea is to run the coreset construction on the rows of $AS$ (which are lower dimensional points). Since the dimensions of $AS$ are $n$ times $k/\varepsilon^2$, 
this has the effect that we can replace $d$ in the coreset size by $O(k/\varepsilon^2)$. Furthermore, we show
that by storing additional information we can still compute an approximate solution for $A$ (the challenge is that
although $AS$ will approximate preserve clustering costs, the cluster centers that achieve this cost lie in a different
space and cannot be used directly as a solution for $A$).

\thmreduction*
\begin{proof}
Let $X$ be the optimal clustering matrix on input $A'S$ and $Y$ be the optimal clustering matrix for
input $A$. Let $Z$ be the clustering matrix returned by our $(\alpha,\beta)$-fair approximation algorithm on input $A'S$ (or, equivalently, 
on input $T$). Let $\varepsilon' = \varepsilon/25$.
Since we are using a $\gamma$-approximation algorithm, we know that $\|ZZ^TA'S - A'S\|_F^2 \le\gamma \cdot \|XX^TA'S-A'S\|_F^2$.
We also observe that $\|ZZ^TA'S - ZZ^TAS\|_F \le \|ZZ^T\|_2 \|A'S-AS|_F = \|A'S -AS\|_F$ for an orthogonal projection matrix $ZZ^T$.
Furthermore, we will use that $\|XX^TA'S - A'S\|_F \le \|XX^T A'S - XX^T AS\|_F + \|XX^TAS - AS\|_F + \|A'S-AS\|_F$ and the fact that the spectral norm and Frobenius norm are conforming, e.g., they satisfy $\|AB\|_F \leq \|A\|_2\|B\|_F$.
We obtain
\begin{eqnarray*}
& &(1-\varepsilon')\cdot \|ZZ^TA-A\|_F^2 
\leq  \|ZZ^TAS-AS\|_F^2  \\
&\leq & \left(\|ZZ^TA'S-ZZ^TAS\|_F +\|ZZ^TA'S-AS\|_F \right)^2  \\
& \leq & (\|ZZ^TA'S - ZZ^TAS\|_F + \|ZZ^TA'S - A'S\|_F \\
& & + \|A'S-AS\|_F)^2 \\
&\leq & \left(2\|A'S-AS\|_F + \sqrt{\gamma} \|XX^TA'S-A'S\|_F\right)^2  \\
& \leq & ( (2+\sqrt{\gamma}) \|A'S-AS\|_F  \\ 
& &  + \sqrt{\gamma} ( \|XX^T A'S - XX^T AS\|_F  + \\ 
& & \|XX^TAS - AS\|_F))^2 \\
& \leq & ((2 + 2\sqrt{\gamma}) \|A'S-AS\|_F +\sqrt{\gamma}\|XX^TAS - AS\|_F)^2\\
& \leq & ((\frac{\varepsilon'}{4} ( 2+2\sqrt{\gamma}) + \sqrt{\gamma}) \|XX^TAS - AS\|_F)^2 \\
& \leq & ((1+\varepsilon') \sqrt\gamma)\|XX^TAS - AS\|_F)^2 \\
& \leq & (1+\varepsilon')^2 \gamma \|YY^TAS -AS\|_F^2\\
& \leq & (1+\varepsilon')^3 \gamma \|YY^T A-A\|_F^2
\end{eqnarray*}
where the first and the last inequality follows from the guarantee of Definition~\ref{def:sketch} and Theorem 12 of~\cite{CohenEMMP15}.
To conclude the proof, observe that $(1+\varepsilon')^3 /{1-\varepsilon} \leq (1+25\varepsilon') = (1+\varepsilon)$ for $0<\varepsilon<\frac{1}{2}$. 
\end{proof}

\section{Acknowledgements}
Chris Schwiegelshohn acknowledges the support of the ERC Advanced Grant 788893 AMDROMA.

\bibliographystyle{amsalpha}

\bibliography{arxiv-references}

\end{document}